\documentclass[10pt,journal,twocolumn,twoside]{IEEEtran}
\ifCLASSINFOpdf
   \usepackage[pdftex]{graphicx}
\else
\fi

\usepackage{cite}
\usepackage{amsmath,amssymb,amsfonts}
\usepackage{textcomp}
\usepackage{subfigure}
\usepackage{stfloats}
\usepackage{multirow}
\usepackage{algorithm}
\usepackage{algorithmic}
\usepackage{graphicx}
\usepackage{setspace}
\usepackage{color}
\usepackage{bm}
\usepackage{mathrsfs}
\usepackage{ntheorem} 
\newtheorem{definition}{Definition} 
\newtheorem{assumption}{Assumption} 
\newtheorem{theorem}{Theorem} 
 
\newtheorem{Proposition}{Proposition}

\newenvironment{proof}{{\indent \indent \it Proof:}}{\hfill $\square$\par}
\newenvironment{prf}{{\noindent \noindent \it }}{\hfill $\square$\par}
\begin{document}
\title{A Multi-Dimensional Matrix Pencil-Based \\
Channel Prediction Method for \\
Massive MIMO with Mobility}
\author{Weidong~Li,~Haifan~Yin,~Ziao~Qin,~Yandi~Cao,
        and~M\'{e}rouane~Debbah,~\IEEEmembership{Fellow,~IEEE}
\thanks{W. Li, H. Yin, Z. Qin and Y. Cao are with Huazhong University of Science and Technology, 430074 Wuhan, China (e-mail: weidongli@hust.edu.cn, yin@hust.edu.cn, ziao\_qin@hust.edu.cn, yandicao@hust.edu.cn).}
\thanks{M. Debbah is with the Technology Innovation Institute, and also with the Mohamed Bin Zayed University of Artificial Intelligence, 9639 Masdar City, Abu Dhabi, United Arab Emirates (e-mail: merouane.debbah@tii.ae).}
\thanks{A part of this work \cite{Li22ICC} was presented in the IEEE International Conference on Communications (IEEE ICC 2022).}
\thanks{This work is supported by the National Natural Science Foundation of China under Grant 62071191.}}
\maketitle
\begin{abstract}
This paper addresses the mobility problem in massive multiple-input multiple-output systems, which leads to significant performance losses in the practical deployment of the fifth generation mobile communication networks. We propose a novel channel prediction method based on multi-dimensional matrix pencil (MDMP), which estimates the path parameters by exploiting the angular-frequency-domain and angular-time-domain structures of the wideband channel. The MDMP method also entails a novel path pairing scheme to pair the delay and Doppler, based on the super-resolution property of the angle estimation. 
Our method is able to deal with the realistic constraint of time-varying path delays introduced by user movements, which has not been considered so far in the literature. 
We prove theoretically that in the scenario with time-varying path delays, the prediction error converges to zero with the increasing number of the base station (BS) antennas, providing that only two arbitrary channel samples are known. 
We also derive a lower-bound of the number of the BS antennas to achieve a satisfactory performance.
Simulation results under the industrial channel model of 3GPP demonstrate that our proposed MDMP method approaches the performance of the stationary scenario even when the users' velocity reaches 120 km/h and the latency of the channel state information is as large as 16 ms.
\end{abstract}


%
\IEEEpeerreviewmaketitle

\section{Introduction}
%
%
%
%
\IEEEPARstart{M}{assive} multiple-input multiple-output (MIMO) technology is playing a key role in enhancing the spectral efficiency of the fifth generation (5G) mobile communication systems \cite{Marzetta10TCom}. Compared to the conventional MIMO, massive MIMO greatly improves the capacity and reliability of mobile communication system by deploying more antenna elements at the base station (BS) \cite{Heath14Mag}.

\par 
Theoretically, the performance of massive MIMO is governed by the accuracy of the channel state information (CSI). In practice, there are several causes of inaccurate CSI, including pilot contamination \cite{jose2011TWC}, imperfect CSI feedback in frequency division duplexing (FDD) mode \cite{Marzetta10TCom}, the mobility problem \cite{Yin20JSAC} (or the ``curse of mobility"), etc. 
Although a rich body of literature has addressed the problems of pilot contamination and CSI feedback in FDD, e.g., in \cite{Yin13JSAC, Muller14JSTSP, Adhikary2013, Jiang15TWC}, the mobility problem has received relatively little attention so far, and is still a challenging problem to be solved. Such a problem mainly results from the mobility of user equipment (UE) and the delay of CSI \cite{Yin20JSAC}. 
The movement of the UE makes the estimation of CSI outdated and unusable for multi-user precoding, particularly in high-mobility scenario with large CSI delay. 
Some papers have focused on the effect of CSI delay and prove that it is harmful for the spectral efficiency performance of massive MIMO \cite{Heath13JCN,Ai21TWC,Larsson18TWC,Papazafeiropoulos15TCom}. 

One way to overcome this problem is by channel prediction, which has been investigated in the literature. 
The work in \cite{18GaoTWC} proposes a spatial-temproal basis expansion model (ST-BEM) method to predict the downlink (DL) channel. 
The authors of \cite{19OgawaTVT} propose a compressed sensing channel prediction method. 
An efficient approximated maximum likelihood estimator is proposed in \cite{Caire20TWC}.
However, the aforementioned literature generally assumes the channel parameters are time-invariant, which might be questionable in practical mobility scenarios. 
The time-varying delay in Wi-Fi communication scenario is studied in \cite{12TomasTVT}, where the authors propose an iterative channel estimator basing on a two-dimensional (2-D) subspace spanned by the discrete prolate spheroidal (DPS) sequence. The DPS sequence has less spectrum leakage than the discrete Fourier transform (DFT) sequence.
However, the DPS sequence needs to meet two specific requirements in applications: ${\omega _{\max }}{T} \ll 1$ and ${\tau _{\max }}\Delta f \ll 1$, where ${\omega _{\max }}$ is the maximum Doppler, ${\tau _{\max }}$ denotes the maximum delay, ${T}$ is the duration of sample, and $\Delta f$ represents the sub-carrier spacing. These requirements might not be easily full-filled in practice.

\par Recently, machine-learning (ML) and neural network (NN) algorithms, such as recurrent NN, conventional NN  and deep NN, have been applied to predict the channel through the trained network with the historical data \cite{19ZhuTCom,20YuanTWC,20GesbertTWC,21GaoJSAC}. The authors in \cite{21GaoJSAC} propose a data-driven channel prediction method by adopting complex-valued NN algorithm. 
However, these algorithms usually need enormous data, which cannot be collected easily in realistic systems. Additionally, these algorithms may perform well in static and low-mobility scenarios, yet they might not achieve the expected performance in high-mobility scenario due to the long training time.
Moreover, the trained network may not perform ideally in the varying environment, due to the challenge of network generalization. 

By transforming the wideband channel into angular-delay domain, the paper \cite{Yin20JSAC} proposes a Prony-based angular-delay domain (PAD) prediction method and proves that it can overcome the problem of CSI aging. However, it does not take the effect of time-varying path delay into consideration and assumes the CSI delay to be an integral multiple of the pilot interval. 
To the best of our knowledge, the realistic effect of time-varying path delay is overlooked in most of the existing literature. Such an effect makes the path delay difficult to track and leads to the distortion of the traditional time-domain multipath channel structure. 


\par In order to break the limitations above and tackle the mobility problem, we propose a novel matrix pencil (MP) based super-resolution channel prediction method. The super-resolution property of our method comes from the super-accurate parameters estimated by MP algorithm.
In the literature, some papers adopt MP algorithm to estimate channel parameters. In \cite{92TSP} the authors adopt 2-D MP algorithm to estimate 2-D frequencies. The work in \cite{15TWC} estimates the path angle and delay in the IEEE 802.11ac context. The elevation angle of departure (EOD), azimuth angle of departure (AOD) and delay are estimated in \cite{19DSP}. 
Compared to other traditional super-resolution algorithms, e.g., multiple signal classification (MUSIC) \cite{86SchmidtTAP} and estimation of signal parameters via rational invariance techniques (ESPRIT) \cite{89RoyTASSP}, MP algorithm has some distinct advantages: it does not search and find peak value in space, and it needs less channel samples.
However, the traditional MP algorithm cannot be directly applied to solve the mobility problem of massive MIMO, because of the time-varying path delay and the multi-dimensional structure of the wideband massive MIMO channel. 


\par In order to achieve more accurate channel predictions in this paper, we extend the traditional MP method to multi-dimensional algorithm that extracts the EOD, AOD, time-varying delay and Doppler simultaneously. 
Specifically, we first introduce a three-dimensional (3-D) MP algorithm to estimate the EOD, AOD, and delay on the angular-frequency domain, and then  put forward another 3-D MP algorithm to estimate the EOD, AOD and Doppler on the angular-time domain. 
Based on the super-resolution property of the estimated EOD and AOD, we propose a pairing algorithm to pair the path delay and Doppler. 
The future CSI is reconstructed with the estimated parameters. 
To the best of our knowledge, our proposed multi-dimensional matrix pencil (MDMP) method is the first attempt to estimate the EODs, AODs, Doppler and the time-varying delays of multiple paths simultaneously.

\par The contributions of this paper are summarized as follows:
\begin{itemize}
\item We propose a super-resolution MDMP prediction method to address the mobility problem with time-varying path delays, which estimates the multipath EODs, AODs, delays and Doppler simultaneously through angular-time-domain and angular-frequency-domain structures of the channel.
Compared to the traditional methods, significant gains of our proposed method are confirmed in simulations.
\item We prove that the prediction error of the MDMP method converges to zero in mobility scenario with arbitrary CSI delay, when the number of the BS antennas and the bandwidth are large enough.
Our proposed method breaks the limitation of the PAD method whose CSI delay is an integral multiple of the pilot interval.

\item We also prove that the prediction error converges to zero providing that only two arbitrary samples are known, when the number of the BS antennas is large enough. Our proposed MDMP method does not necessarily require the samples to be neighboring ones, yet it is a common assumption in most existing works.
\item We derive a lower-bound of the number of the BS antennas in a wideband channel to give a satisfactory performance of the MDMP method. 
The lower-bound is correlated with the number of samples. 
As the number of samples increases, the lower-bound decreases. 

\end{itemize}

\par This paper is organized as follows: Sec. \ref{sec:system model} introduces the channel model. In Sec. \ref{sec:MDMP method}, MDMP prediction method is proposed. The performance of MDMP method is analyzed in Sec. \ref{sec:performance analysis}. The simulation results are shown in Sec. \ref{sec:simulation}, and Sec. \ref{sec:conclusion} concludes the paper. 

\par Notations: We use boldface to represent vectors and matrices, where ${{\bf{0}}_{{M_1} \times {N_1}}}$, ${{\bf{I}}_{{M_2}}}$, ${\Upsilon _{{M_3}}}$ and ${{\bf{1}}_{{M_4} \times {N_5}}}$ denote ${M_1} \times {N_1}$ zero matrix, ${M_2} \times {M_2}$ identity matrix, ${M_3} \times {M_3}$ anti-identity matrix and ${M_4} \times {N_5}$ all-ones matrix. 
Let ${({\bf{X}})^T}$, ${({\bf{X}})^*}$, ${({\bf{X}})^H}$, ${({\bf{X}})^{ - 1}}$ and ${({\bf{X}})^\dag }$ denote the transpose, conjugate, conjugate transpose, inverse and Moore-Penrose pseudoinverse of the matrix ${\bf{X}}$, respectively. ${\rm{card}}(\cdot )$ denotes the number of the elements in a set. $r\{  \cdot \} $ denotes the rank of a matrix. 
${\left\|  \cdot  \right\|_F}$ stands for the Frobenius norm. 
${\mathop{\rm Re}\nolimits} \left\{  \cdot  \right\}$ and ${\mathop{\rm Im}\nolimits} \left\{  \cdot  \right\}$ take the real and imaginary components of a complex number. $E\{  \cdot \} $ is the expectation operation. ${\rm{diag}}\{  \cdot \} $ denotes the diagonal operation of a matrix. $[{\bf{X}}:{\bf{Y}}]$ is the extending matrix of ${\bf{X}}$ and ${\bf{Y}}$. 
${\bf{X}} \odot {\bf{Y}}$ and ${\bf{X}} \otimes {\bf{Y}}$ denote the Hadamard product and Kronecker product of ${\bf{X}}$ and ${\bf{Y}}$.

\section{Channel Model}\label{sec:system model}
\par We consider a wideband time division duplexed (TDD) massive MIMO system, where a BS serves multiple UEs. In such a system, the BS estimates the CSI from the uplink (UL) pilot, and the DL CSI is acquired based on channel reciprocity. 
\begin{figure}[htb]
\centering
\includegraphics[width=3.5in]{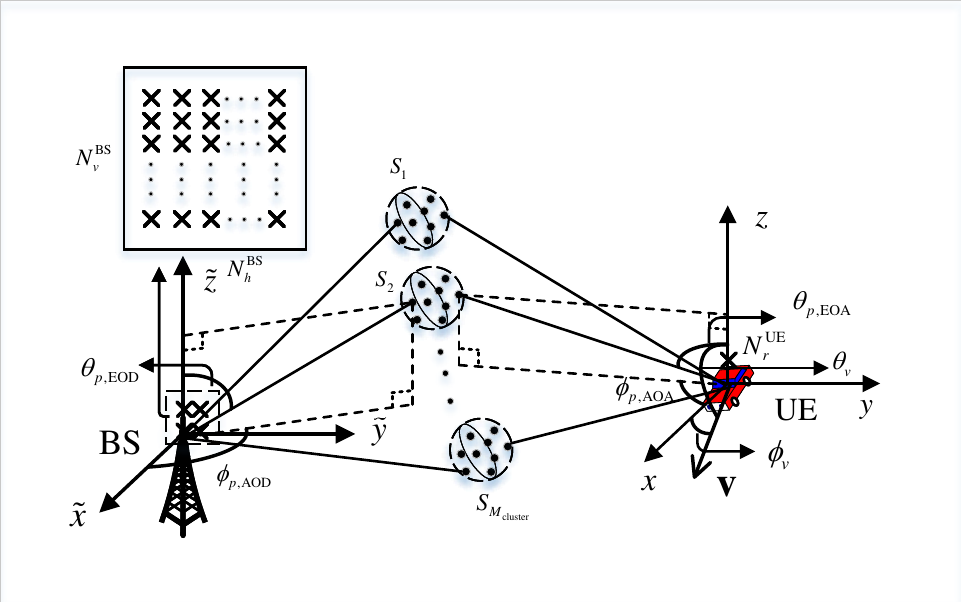}
\caption{The multipath channel between the BS and the UE.}
\vspace{-0.3cm}
\label{figure_speed_se}
\end{figure}
\par Fig. \ref{figure_speed_se} illustrates the channel between the BS and the UE. In most 5G commercial systems, the BS is equipped with a uniform planar array (UPA), which consists of $N_h^{{\rm{BS}}}$ columns and $N_v^{{\rm{BS}}}$ rows of antenna elements.
The numbers of the BS antennas and the UE antennas are respectively $N_t={N_h^{{\rm{BS}}}}{N_v^{{\rm{BS}}}}$ and $N_r^{{\rm{UE}}}$.
There are ${M_{\rm{cluster}}}$ different clusters, which are denoted by ${S_i}$, $i = 1,2, \cdots ,{M_{\rm{cluster}}}$. Each cluster contains many scattering rays.
The UL and DL channels share the same bandwidth, which is composed of ${N_f}$ subcarriers with spacing $\Delta f$.
\par The DL MIMO channel between the BS and the UE is expressed as ${\bf{H}}(t,f) = {\left[ {{h_{u,{s_h},{s_v}}}(t,f)} \right]_{N_r^{{\rm{UE}}} \times N_h^{{\rm{BS}}}N_v^{{\rm{BS}}}}}$, where ${h_{u,{s_h},{s_v}}}(t,f)$ denotes the channel frequency response between the ${s_h}$-th column and the ${s_v}$-th row  of the BS antenna array and the $u$-th antenna of the UE, and is modeled as \cite{3GPPR16}
\begin{equation}
\begin{array}{l}
\!\!\!{h_{u,{s_h},{s_v}}}(t,f)\\
\!\!\!= \sum\limits_{p = 1}^P {} {\beta _p}{e^{\frac{{j2\pi ({\bf{\hat r}}{{_p^{{\rm{rx}}}})^T}{\bf{\bar d}}_u^{{\rm{rx}}}}}{{{\lambda _0}}}}}{e^{\frac{{j2\pi ({\bf{\hat r}}{{_p^{{\rm{tx}}}})^T}{\bf{\bar d}}_{{s_h},{s_v}}^{{\rm{tx}}}}}{{{\lambda _0}}}}}{e^{j2\pi {\omega _p}t}}{e^{-j2\pi f{\tau _p}(t)}}
\end{array},\!\!
\!\label{ChannelModel}
\end{equation}
where $P$ denotes the number of paths, and ${\beta _p}$ is the complex amplitude of the $p$-th path. Also, ${\lambda _0} = \frac{{\rm{c}}}{{{f_c}}}$ is the wavelength, where ${f_c}$ is the central carrier frequency and ${\rm{c}}$ is the speed of light. Furthermore, ${\bf{\hat r}}_p^{{\rm{rx}}}$ and ${\bf{\hat r}}_p^{{\rm{tx}}}$ denote the spherical unit vectors of the UE and the BS, and are expressed as
\begin{equation}
\begin{array}{l}
{\bf{\hat r}}_p^{{\rm{rx}}} = \left[ \begin{array}{l}
\cos {\theta _{p,{\rm{EOA}}}}\cos {\phi _{p,{\rm{AOA}}}}\\
\cos {\theta _{p,{\rm{EOA}}}}\sin {\phi _{p,{\rm{AOA}}}}\\
\ \ \ \ \ \ \ \sin {\theta _{p,{\rm{EOA}}}}
\end{array} \right],
\end{array}
\!\label{R_rx}
\end{equation}
\begin{equation}
\begin{array}{l}
{\bf{\hat r}}_p^{{\rm{tx}}} = \left[ \begin{array}{l}
\cos {\theta _{p,{\rm{EOD}}}}\cos {\phi _{p,{\rm{AOD}}}}\\
\cos {\theta _{p,{\rm{EOD}}}}\sin {\phi _{p,{\rm{AOD}}}}\\
\ \ \ \ \ \ \ \sin {\theta _{p,{\rm{EOD}}}}
\end{array} \right],
\end{array}
\!\label{R_tx}
\end{equation}
where ${\theta _{p,{\rm{EOA}}}}$, ${\phi _{p,{\rm{AOA}}}}$, ${\theta _{p,{\rm{EOD}}}}$ and ${\phi _{p,{\rm{AOD}}}}$ are the elevation angle of arrival (EOA), azimuth angle of arrival (AOA), EOD and AOD, respectively. Additionally, ${\bf{\bar d}}_{{s_h}{s_v}}^{{\rm{tx}}}$ and ${\bf{\bar d}}_u^{{\rm{rx}}}$ are the location vectors of the BS and the UE antennas:
\begin{equation}
{\bf{\bar d}}_{{s_h},{s_v}}^{{\rm{tx}}} = {\left[ {\begin{array}{*{20}{c}}{0,}&{d_h^{{\rm{tx}}}({s_h} - 1),}&{d_v^{{\rm{tx}}}({s_v} - 1)} \end{array}} \right]^T},
\!\label{Vector_BS}
\end{equation}
where $d_h^{{\rm{tx}}}$ is the horizontal antenna spacing, and $d_v^{{\rm{tx}}}$ is the vertical antenna spacing. Moreover, ${\omega _p} = \frac{{ ({\bf{\hat r}}{{_p^{{\rm{rx}}}})^T}{\bf{v}}}}{{{\lambda _0}}}$ represents the Doppler, and ${\bf{v}}$ is the velocity vector of the UE:
\begin{equation}
{\bf{v}} = v{\left[ {\begin{array}{*{20}{c}}{\cos {\theta _v}\cos {\phi _v},}&{\cos {\theta _v}\sin {\phi _v},}&{\sin {\theta _v}} \end{array}} \right]^T},
\!\label{Velocity}
\end{equation}
where $v$ is the speed, ${\theta _v}$ is the elevation angle of velocity, and ${\phi _v}$ is the azimuth angle of velocity. The $p$-th path delay ${\tau _p}(t)$ is time-varying and modeled as \cite{3GPPR16}
\begin{equation}
{\tau _p}(t) = {\tau _{p,0}} + {k_{{\tau _p}}} t = {\tau _{p,0}} - \frac{{({\bf{\hat r}}{{_p^{{\rm{rx}}}})^T}{\bf{v}}}}{{\rm{c}}} t = {\tau _{p,0}} - \frac{{{\omega _p}}}{{{f_c}}} t,
\!\label{Delay}
\end{equation}
where ${\tau _{p,0}}$ is the initial value, and ${k_{{\tau _p}}}$ denotes the changing rate of the path delay. Notice that, besides the time domain, Doppler also has an effect on the frequency domain.
The 3-D steering vector ${{\bf{a}}^{{\rm{tx}}}}({\theta _{p}},{\phi _{p}})$ is expressed as
\begin{equation}
\begin{array}{l}
{{\bf{a}}^{{\rm{tx}}}}({\theta _{p}},{\phi _{p}})= {\left[ {\begin{array}{*{20}{c}}{{(a_h^{{\rm{tx}}})}^0,}&{{(a_h^{{\rm{tx}}})}^1,}&{\cdots,}&{{(a_h^{{\rm{tx}}})}^{(N_h^{{\rm{BS}}} - 1)}} \end{array}} \right]^T}\\
\ \ \ \ \ \ \ \ \ \ \ \ \ \ \otimes {\left[ {\begin{array}{*{20}{c}}{{(a_v^{{\rm{tx}}})}^0,}&{{(a_v^{{\rm{tx}}})}^1,}&{\cdots,}&{{(a_v^{{\rm{tx}}})}^{(N_v^{{\rm{BS}}} - 1)}} \end{array}} \right]^T},
\end{array}
\!\label{SteeringVector}
\end{equation}
where ${\theta _p}$ and ${\phi _p}$ are  ${\theta _{p,{\rm{EOD}}}}$ and ${\phi _{p,{\rm{AOD}}}}$ for simplicity.
Furthermore, $a_h^{{\rm{tx}}}$ and $a_v^{{\rm{tx}}}$ denote the spatial signatures in the directions of $\cos ({\theta _{p}})\sin({\phi _{p}})$ and $\sin ({\theta _{p}})$, which are expressed as
\begin{equation}
a_h^{{\rm{tx}}} = {\rm{exp}}({\frac{{j2\pi d_h^{{\rm{tx}}}\cos ({\theta _p})\sin({\phi _p})}}{{{\lambda _0}}}}),
\!\label{ah}
\end{equation}
\begin{equation}
a_v^{{\rm{tx}}} = {\rm{exp}}({\frac{{j2\pi d_v^{{\rm{tx}}}\sin ({\theta _p})}}{{{\lambda _0}}}}).
\!\label{av}
\end{equation}
Let ${{\bf{h}}_u}(t,f)$ denote the channel between all BS antennas and the $u$-th UE antenna at time $t$ and frequency $f$:
\begin{equation}
\begin{array}{l}
\!\!\!\!\!\!\!\!\!{{\bf{h}}_u}(t,f) \!=\!\! \sum\limits_{p = 1}^P {} \!\!{\beta _p}{e^{\frac{{j2\pi ({\bf{\hat r}}{{_p^{{\rm{rx}}}}})^T{\bf{\bar d}}_u^{{\rm{rx}}}}}{{{\lambda _0}}}}}\!\!{e^{j2\pi {\omega _p}t}}{e^{-j2\pi f{\tau _p}(t)}}{{\bf{a}}^{{\rm{tx}}}}({\theta _{p}},{\phi _{p}}).\!\!
\!\!\!\end{array}
\!\!\!\label{Channel_UE_frequency}
\end{equation}
In the time domain, the channels at all subcarriers are expressed as 
\begin{equation}
\begin{array}{l}
\!\!\!\!\!\!\!\!\!{{\bf{H}}_u}(t) = [{\begin{array}{*{20}{c}}{{\bf{h}}_u^T(t,{f_1}),}&{{\bf{h}}_u^T(t,{f_2}),}&{\cdots,}&{{\bf{h}}_u^T(t,{f_{{N_f}}})} \end{array}}],
\!\!\!\end{array}
\!\!\!\label{Channel_UE_all subcarriers}
\end{equation}
where ${f_{{n_f}}}$ is the frequency of the ${n_f}$-th subcarrier:
\begin{equation}
{f_{{n_f}}} = {f_1} + ({n_f} - 1)\Delta f,{\rm{ }}{n_f} = 1,2, \cdots {N_f}.
\!\label{Frequency}
\end{equation}
Based on the angular-frequency-domain channel structure, ${{\bf{H}}_u}(t)$ is rewritten as
\begin{equation}
{{\bf{H}}_u}(t) = {\bf{A}}_u^{{\rm{tx}}}{{\bf{C}}_u}{{\bf{B}}_u},
\!\label{Channel_UE_angular-frequency}
\end{equation}
where ${\bf{A}}_u^{{\rm{tx}}}$ contains the 3-D steering vectors of all paths:
\begin{equation}
\begin{array}{l}
\!\!\!\!\!\!\!\!\!\!{\bf{A}}_u^{{\rm{tx}}} = [{\begin{array}{*{20}{c}}{{{\bf{a}}^{{\rm{tx}}}}({\theta _{1}},\!{\phi _{1}}),}&{\cdots,}&{{{\bf{a}}^{{\rm{tx}}}}({\theta _{P}},\!{\phi _{P}})} \end{array}}].
\!\!\!\!\!\!\!\end{array}
\!\!\!\!\!\!\label{Steering_vector_EOD-AOD}
\end{equation}
The frequency-domain matrix ${{\bf{B}}_u}$ is expressed as
\begin{equation}
{{\bf{B}}_u} = \left[ {\begin{array}{*{20}{c}}{{\bf{b}}({f_1}),}&{{\bf{b}}({f_2}),}&{\cdots,}&{{\bf{b}}({f_{{N_f}}})} \end{array}} \right],
\!\label{Delay_matrix}
\end{equation}
where ${\bf{b}}({f_{{n_f}}})$ denotes the delay response vector at the ${n_f}$-th subcarrier:
\begin{equation}
{\bf{b}}({f_{{n_f}}}) = {\left[ {\begin{array}{*{20}{c}}{e^{ - j2\pi {f_{{n_f}}} {\tau _1}(t)},}&{\cdots,}&{{e^{ - j2\pi {f_{{n_f}}} {\tau _P}(t)}}} \end{array}} \right]^T}.\!\!
\!\!\label{Delay_vector}
\end{equation}
The diagonal matrix ${{\bf{C}}_u}$ is expressed as
\begin{equation}
\begin{array}{l}
\!\!\!\!\!\!\!\!{{\bf{C}}_u}\\
\!\!\!\!\!\!\!=\! {\rm{diag}}\!\!\left\{\! {{\beta _1}{e^{\frac{{j2\pi ({\bf{\hat r}}{{_1^{{\rm{rx}}}}})^T{\bf{\bar d}}_u^{{\rm{rx}}}}}{{{\lambda _0}}}}}\!{e^{j2\pi {\omega _1}t}}, \!\cdots \!,{\beta _P}{e^{\frac{{j2\pi ({\bf{\hat r}}{{_P^{{\rm{rx}}}}})^T{\bf{\bar d}}_u^{{\rm{rx}}}}}{{{\lambda _0}}}}}\!{e^{j2\pi {\omega _P}t}}} \!\right\}\!\!.
\!\!\end{array}
\!\!\!\!\!\!\label{Diagonal_Cu}
\end{equation}
According to the angular-time-domain channel structure, ${{\bf{H}}_u}(t)$ is also rewritten as 
\begin{equation}
{{\bf{H}}_u}(t) = {\bf{A}}_u^{{\rm{tx}}}{{\bf{D}}_u}({{\bf{E}}_u} \odot {{\bf{F}}_u}),
\!\!\label{Channel_UE_angular-time}
\end{equation}
where ${{\bf{D}}_u}$ is a diagonal matrix:
\begin{equation}
{{\bf{D}}_u} = {\rm{diag}}\left\{ {\begin{array}{*{20}{c}}{{\beta _1}{e^{\frac{{j2\pi ({\bf{\hat r}}{{_1^{{\rm{rx}}}}})^T{\bf{\bar d}}_u^{{\rm{rx}}}}}{{{\lambda _0}}}}},}&{\cdots,}&{{\beta _P}{e^{\frac{{j2\pi ({\bf{\hat r}}{{_P^{{\rm{rx}}}}})^T{\bf{\bar d}}_u^{{\rm{rx}}}}}{{{\lambda _0}}}}}} \end{array}} \right\}.
\!\!\label{Diagonal_Du}
\end{equation}
The matrix ${{\bf{E}}_u}$ is composed of the initial delay response vectors at all subcarriers:
\begin{equation}
{{\bf{E}}_u} = \left[ {\begin{array}{*{20}{c}}{{\bf{e}}({f_1}),}&{{\bf{e}}({f_2}),}&{\cdots,}&{{\bf{e}}({f_{{N_f}}})} \end{array}} \right],
\!\!\label{Delay_initial_matrix}
\end{equation}
where ${\bf{e}}({f_{{n_f}}})$ is for the ${n_f}$-th subcarrier:
\begin{equation}
{\bf{e}}({f_{{n_f}}}) = {\left[ {\begin{array}{*{20}{c}}{e^{ - j2\pi {f_{{n_f}}} {\tau _{1,0}}},}&{\cdots,}&{e^{ - j2\pi {f_{{n_f}}} {\tau _{P,0}}}} \end{array}} \right]^T}.
\!\!\label{Delay_initial_vector}
\end{equation}
The time-domain matrix ${{\bf{F}}_u}$ is expressed as
\begin{equation}
{{\bf{F}}_u} = {\left[ {\begin{array}{*{20}{c}}{{\bf{f}}({f_1}),}&{{\bf{f}}({f_2}),}&{\cdots,}&{{\bf{f}}({f_{{N_f}}})} \end{array}} \right]^T},
\!\!\label{Time_matrix}
\end{equation}
where ${\bf{f}}({f_{{n_f}}}), 1 \leq n_f \leq N_f$, is defined as:
\begin{equation}
\!{\bf{f}}({f_{{n_f}}}) = {\left[ {\begin{array}{*{20}{c}}{e^{j2\pi ({\omega _1} - {f_{{n_f}}} {k_{{\tau _1}}})t},}&{\cdots,}&{e^{j2\pi ({\omega _P} - {f_{{n_f}}} {k_{{\tau _P}}})t}} \end{array}} \right]^T}\!\!\!.\!\!\!
\!\!\!\label{Time_vector}
\end{equation}
Notice that Eq. (\ref{Channel_UE_angular-time}) is different from the result in \cite{Yin20JSAC}, as we simultaneously consider the effect of Doppler on time domain and frequency domain.
The multipath delays, angles and Doppler will be estimated jointly by our proposed MDMP method in the next section.

\section{The Proposed MDMP Prediction Method}\label{sec:MDMP method}
\par In this section, we introduce our proposed MDMP super-resolution method for channel prediction. In general, our approach first estimates the EOD, AOD, delay, and Doppler information by exploiting the angular-frequency-domain and angular-time-domain structures of the channel, and then performs a path pairing procedure to determine the corresponding parameters of the paths. Finally the BS predicts the CSI by reconstructing the future channel with the estimated parameters.
To ease the exposition, we briefly introduce here the derivations in Sec. \ref{sec:angle-delay estimation} and Sec. \ref{sec:angle-Doppler estimation}. Since Sec. \ref{sec:angle-delay estimation} and Sec. \ref{sec:angle-Doppler estimation} have similar derivations, we take Sec. \ref{sec:angle-delay estimation} for example. The derivation mainly contains five steps, i.e., generating the 3-D MP matrix from the channels, transforming the 3-D complex MP matrix to a real matrix, determining the number of paths, estimating the path delays, and estimating the EODs and the AODs.

\subsection{The 3-D Angle-Delay Estimation}\label{sec:angle-delay estimation}
\par We first generate a 3-D MP matrix. For ease of exposition, we start with the one-dimensional (1-D) setting of the matrix pencil method, and then move on to the 2-D and 3-D cases.
More specifically, we first generate a 1-D MP matrix only containing the information of $\cos {\theta _p}\sin {\phi _p}$, by windowing the BS antenna panel in the horizontal direction. Then, we generate a 2-D MP matrix by adding the second window in the vertical direction of the BS antenna panel, which also contains the information of $\sin {\theta _p}$. Finally, based on the 2-D MP matrix, we generate a 3-D MP matrix by adding the third window in the frequency domain, which contains the information of path delay. More details will be shown below. 

\par Based on the angular-frequency-domain channel structure, a 1-D MP matrix is generated by sliding a horizontal window within the BS antenna panel. 
Such a MP matrix ${{\bf{G}}_{u,r}}(t,{n_f}) \in {{\mathbb{C}}^{L \times (N_h^{{\rm{BS}}} - L + 1)}}$ is defined as
\begin{equation}
{{\bf{G}}_{u,r}}(t,{n_f}) \!=\!\! \left[\!\! \begin{array}{l}
{h_{u,1,r}}(t,{n_f}), \cdots \!,{h_{u,(N_h^{{\rm{BS}}} \!-\! L \!+\! 1),r}}(t,{n_f})\\
 \ \ \ \ \ \ \vdots \ \ \ \ \ \ \ \ \ \ddots \ \ \ \ \ \ \ \ \vdots \\
{h_{u,L,r}}(t,{n_f}), \cdots ,\ \ {h_{u,N_h^{{\rm{BS}}},r}}(t,{n_f})
\end{array} \!\!\!\right]\!\!,
\!\!\!\!\label{1DMP}
\end{equation}
which is composed of the channels between the $u$-th UE antenna and all horizontal antennas in the $r$-th row of the BS antenna panel at the ${n_f}$-th subcarrier, and $L$ is the pencil size that satisfies $P < L < N_h^{{\rm{BS}}} -P+1$. 
Based on Eq. (\ref{Channel_UE_angular-frequency}), the 1-D MP matrix is rewritten as 
\begin{equation}
{{\bf{G}}_{u,r}}(t,{n_f}) = {{\bf{X}}_{\theta ,\phi ,1}}{{\bf{Z}}_{\theta}^{r}}{\bf{Y}}{{\bf{Z}}_{\tau}^{{n_f}-1}}{{\bf{X}}_{\theta ,\phi ,2}},
\!\label{1D-MP-Gu}
\end{equation}
where the $p$-th column of ${{\bf{X}}_{\theta ,\phi ,1}} \in {{\mathbb{C}}^{L \times P}}$ and the $p$-th row of ${{\bf{X}}_{\theta ,\phi ,2}} \in {{\mathbb{C}}^{P \times (N_h^{{\rm{BS}}} - L + 1)}}$, are the steering vectors of two subsets of antenna array with size $1 \times L$ and $1 \times (N_h^{{\rm{BS}}} - L + 1)$ in the BS antenna panel, respectively. The two matrices ${{\bf{X}}_{\theta ,\phi ,1}}$ and ${{\bf{X}}_{\theta ,\phi ,2}}$ are defined as
\begin{equation}
{{\bf{X}}_{\theta ,\phi ,1}} = \left[ {\begin{array}{*{20}{c}}{{\bf{a}}_{h,L-1}^{{\rm{tx}}}({\theta _{1}},{\phi _{1}}),}&{\cdots,}&{{\bf{a}}_{h,L-1}^{{\rm{tx}}}({\theta _{P}}\!,\!{\phi _{P}})} \end{array}} \right],
\!\!\!\!\label{Matrix_X1}
\end{equation}
\begin{equation}
\!\!\!\!\!\!{{\bf{X}}_{\theta ,\phi ,2}} = {\left[ {\begin{array}{*{20}{c}}{{\bf{a}}_{h,N_h^{{\rm{BS}}} - L}^{{\rm{tx}}}({\theta _{1}},{\phi _{1}}),}&{\cdots,}&{{\bf{a}}_{h,N_h^{{\rm{BS}}} - L}^{{\rm{tx}}}({\theta _{P}},{\phi _{P}})} \end{array}} \!\!\right]^T}\!\!,\!\!\!\!\!\!
\!\!\!\!\label{Matrix_X2}
\end{equation}
where ${\bf{a}}_{h,{l_1}}^{{\rm{tx}}}({\theta _p},{\phi _p}) = {\left[ {{{(a_h^{{\rm{tx}}})}^0},{{(a_h^{{\rm{tx}}})}^1}, \cdots ,{{(a_h^{{\rm{tx}}})}^{{l_1}}}} \right]^T}$, ${l_1} \in \{ L-1,N_h^{{\rm{BS}}} - L\}$. 
The matrix ${{\bf{Z}}_\tau}$ contains the phase differences between two neighboring subcarriers for the paths:
\begin{equation}
{{\bf{Z}}_\tau} = {\rm{diag}}\left\{ {\begin{array}{*{20}{c}}{e^{ - j2\pi \Delta f{\tau _1}(t)},}&{\cdots,}&{e^{ - j2\pi \Delta f{\tau _P}(t)}} \end{array}} \right\}.
\!\label{Diagonal_matrix_delay}
\end{equation}
Likewise, ${{\bf{Z}}_\theta}$ denotes the matrix of the phase differences between two neighboring vertical BS antennas for the paths: 
\begin{equation}
{{\bf{Z}}_\theta} = {\rm{diag}}\left\{ {\begin{array}{*{20}{c}}{e^{\frac{{j2\pi d_v^{{\rm{tx}}}\sin ({\theta _{1}})}}{{{\lambda _0}}}},}&{\cdots,}&{e^{\frac{{j2\pi d_v^{{\rm{tx}}}\sin({\theta _{P}})}}{{{\lambda _0}}}}} \end{array}} \right\}.
\!\label{Diagonal_matrix_angle}
\end{equation}
The matrix ${\bf{Y}}$ is defined as
\begin{equation}
{\bf{Y}} = {\rm{diag}}\left\{ {{{\bf{C}}_u} {\bf{b}}({f_1})} \right\},
\!\label{S_frequency}
\end{equation}
where ${\bf{b}}({f_1})$ and ${{\bf{C}}_u}$ are defined in Eq. (\ref{Delay_vector}) and Eq. (\ref{Diagonal_Cu}).
\par By adding a new window along the vertical direction of the BS antenna panel, a 2-D MP matrix ${{\bf{G}}_u}(t,{n_f}) \in {{\mathbb{C}}^{LR \times (N_h^{{\rm{BS}}} - L + 1)(N_v^{{\rm{BS}}} - R + 1)}}$ is introduced
\begin{equation}
{{\bf{G}}_u}(t,{n_f}) \!=\!\!\! \left[\!\! \begin{array}{l}
{{\bf{G}}_{u,1}}(t,{n_f}), \cdots\! ,\ {{\bf{G}}_{u,(N_v^{{\rm{BS}}}\! - \!R + 1)}}(t,{n_f})\\
 \ \ \ \ \vdots \ \ \ \ \ \ \ \ \ \ \ddots \ \ \ \ \ \ \ \ \ \ \ \vdots \\
{{\bf{G}}_{u,R}}(t,{n_f}), \cdots ,\ \ \ \ {{\bf{G}}_{u,N_v^{{\rm{BS}}}}}(t,{n_f})
\end{array} \!\right]\!,
\!\!\!\label{2D-MP}
\end{equation}
which is composed of the channels between the $u$-th UE antenna and all BS antennas at the ${n_f}$-th subcarrier, and contains the angular information of the paths. The corresponding pencil size is denoted by $R$ that satisfies $P < R < N_v^{{\rm{BS}}}-P+1$. 
Substituting Eq. (\ref{1D-MP-Gu}) into Eq. (\ref{2D-MP}), ${{\bf{G}}_u}(t,{n_f})$ is rewritten as
\begin{equation}
{{\bf{G}}_u}(t,{n_f}) = {{\check{\bf{E}}}_{1}}{\bf{Y}}{{\bf{Z}}_{\tau}^{{n_f}-1}}{{\check{\bf{F}}}_{1}},
\!\label{2D-MP-Gu}
\end{equation}
where the $p$-th column of ${{\check{\bf{E}}}_{1}} \in {{\mathbb{C}}^{LR \times P}}$ and the $p$-th row of ${{\check{\bf{F}}}_{1}} \in {{\mathbb{C}}^{P \times (N_h^{{\rm{BS}}} - L + 1)(N_v^{{\rm{BS}}} - R + 1)}}$ are the 3-D steering vectors of two subsets of antenna array with size $R \times L$ and $(N_v^{{\rm{BS}}} - R + 1) \times (N_h^{{\rm{BS}}} - L + 1)$ on the BS antenna panel, respectively. 
They are defined as
\begin{equation}
{{\check{\bf{E}}}_{1}} = {\left[ {\begin{array}{*{20}{c}}{({{\bf{X}}_{\theta ,\phi ,1}}{{\bf{Z}}_{\theta}^{0}})^T,}&{\cdots,}&{({{\bf{X}}_{\theta ,\phi ,1}}{{\bf{Z}}_{\theta}^{R - 1}})^T} \end{array}} \right]^T},
\!\label{Eleft_1}
\end{equation}
\begin{equation}
{{\check{\bf{F}}}_{1}} = \left[ {\begin{array}{*{20}{c}}{{{\bf{Z}}_{\theta}^0}{{\bf{X}}_{\theta ,\phi ,2}},}&{\cdots,}&{{{\bf{Z}}_{\theta}^{N_v^{{\rm{BS}}} - R}}{{\bf{X}}_{\theta ,\phi ,2}}} \end{array}} \right].
\!\label{Eright_1}
\end{equation}

Adding the third window in the frequency domain, a 3-D MP matrix ${{\bf{G}}_u}(t) \in {{\mathbb{C}}^{LRK \times (N_h^{{\rm{BS}}} - L + 1)(N_v^{{\rm{BS}}} - R + 1)({N_f} - K + 1)}}$ is generated from ${{\bf{G}}_u}(t,{n_f})$ and expressed as
\begin{equation}
{{\bf{G}}_u}(t) \!=\!\! \left[ \!\!\begin{array}{l}
{{\bf{G}}_u}(t,1), \ \cdots ,{{\bf{G}}_u}(t,{N_f} \!-\! K\! +\! 1)\\
 \ \ \ \ \ \vdots \ \ \ \ \ \ \ddots \ \ \ \ \ \ \ \ \ \ \ \vdots \\
{{\bf{G}}_u}(t,K), \cdots ,\ \ \ \ {{\bf{G}}_u}(t,{N_f})
\end{array} \right]\!\!,
\!\!\!\!\!\label{3D-MP}
\end{equation}
where $K$ is the pencil size that satisfies $P < K < {N_f}-P+1$. 
The matrix ${{\bf{G}}_u}(t)$ is composed of the channels between the $u$-th UE antenna and all BS antennas at all subcarriers, and also contains the angular and delay information of the paths.
For notational simplicity, we define $\mu_1=LRK$ and $\mu_2=(N_h^{{\rm{BS}}} - L + 1)(N_v^{{\rm{BS}}} - R + 1)({N_f} - K + 1)$.
According to the channel angular-frequency-domain structure in Eq. (\ref{Channel_UE_angular-frequency}), ${{\bf{G}}_u}(t)$ is rewritten as
\begin{equation}
{{\bf{G}}_u}(t) = {{\check{\bf{E}}}_{2}}{\bf{Y}}{{\check{\bf{F}}}_{2}},
\!\label{3D-MP-Gu}
\end{equation}
where ${{\check{\bf{E}}}_{2}} \in {{\mathbb{C}}^{\mu_1 \times P}}$ and ${{\check{\bf{F}}}_{2}} \in {{\mathbb{C}}^{P \times \mu_2}}$ are defined as
\begin{equation}
{{\check{\bf{E}}}_{2}} = {\left[ {\begin{array}{*{20}{c}}{({{\check{\bf{E}}}_{1}}{{\bf{Z}}_{\tau}^{0}})^T,}&{\cdots,}&{({{\check{\bf{E}}}_{1}}{{\bf{Z}}_{\tau}^{K - 1}})^T} \end{array}} \right]^T},
\!\label{Eleft2}
\end{equation}
\begin{equation}
{{\check{\bf{F}}}_{2}} = \left[ {\begin{array}{*{20}{c}}{{{\bf{Z}}_{\tau}^{0}}{{\check{\bf{F}}}_{1}},}&{\cdots,}&{{{\bf{Z}}_{\tau}^{{N_f} - K}}{{\check{\bf{F}}}_{1}}} \end{array}} \right].
\!\label{Eright2}
\end{equation}
The physical meaning of ${{\check{\bf{E}}}_{1}}{{\bf{Z}}_{\tau}^{{k_1}}},{k_1} = 0, \cdots ,K - 1$ is the EOD-AOD steering vectors of a subset of antenna array with size $R \times L$ on the BS antenna panel at the $({k_1} + 1)$-th subcarrier of the paths. 
Likewise, ${{\bf{Z}}_{\tau}^{{k_2}}}{{\check{\bf{F}}}_{1}},{k_2} = 0, \cdots ,{N_f} - K$ is the EOD-AOD steering vectors of a subset of antenna array with size $(N_v^{{\rm{BS}}} - R + 1) \times (N_h^{{\rm{BS}}} - L + 1)$ on the BS antenna panel at the $({k_2} + 1)$-th subcarrier for the paths.
The 3-D MP matrix ${{\bf{G}}_u}(t)$ contains the information of multipaths, if its rank satisfies $r({{\bf{G}}_u}(t)) = P$.
In this case, the three pencil sizes yield
\begin{equation}
\left\{ \begin{array}{l}
\ \ \ \ \ \ \ \ \ \ \ \ \ \ \ \ \ \ LR(K - 1) \ge P\\
\ \ \ \ \ \ \ \ \ \ \ \ \ \ \ \ \ \ LK(R - 1) \ge P\\
\ \ \ \ \ \ \ \ \ \ \ \ \ \ \ \ \ \ RK(L - 1) \ge P\\
(N_h^{{\rm{BS}}} - L + 1)(N_v^{{\rm{BS}}} - R + 1)({N_f} - K + 1) \ge P
\end{array} \!\!\!\!\right..
\!\!\!\label{NC_3D-MP}
\end{equation}

\par Then, we transform the 3-D MP matrix ${{\bf{G}}_u}(t)$ to a real matrix ${{\bf{G}}_{\rm{re}}}(t) \in {{\mathbb{C}}^{{\mu_1} \times 2{\mu_2}}}$ by the unitary matrix pencil (UMP) method \cite{91KehTSP}, which reduces the computational complexity without losing accuracy. The real matrix ${{\bf{G}}_{\rm{re}}}(t)$ is transformed by 
\begin{equation}
{{\bf{G}}_{\rm{re}}}(t) = {\bf{Q}}_{\mu_1}^H{{\bf{G}}_{\rm{ex}}}(t){{\bf{Q}}_{2{\mu_2}}},
\!\!\!\label{3D-Gre}
\end{equation}
where
\begin{equation}
{{\bf{G}}_{\rm{ex}}}(t) = \left[ {{{\bf{G}}_u}(t):{\Upsilon_{\mu_1}} {{\bf{G}}_u^*(t)}{\Upsilon_{{\mu_2}}} } \right].
\!\!\!\!\!\label{3D-Gex}
\end{equation}
We define ${{\bf{Q}}_{\mu_1 }}$ and ${{\bf{Q}}_{2{\mu_2 }}}$ as two unitary matrices depending on the size of $\mu_1$ and $2{\mu_2}$.
Taking ${{\bf{Q}}_{\mu_1 }}$ for example, if $\mu_1$ is even,
\begin{equation}
{{\bf{Q}}_{\mu_1} } = \frac{1}{{\sqrt 2 }}\left[ \begin{array}{l}
{{\bf{I}}_{\frac{{\mu_1} }{2}}}\ \ \ \ \ j{{\bf{I}}_{\frac{{\mu_1} }{2}}}\\
{\Upsilon _{\frac{{\mu_1} }{2}}}\  -j{\Upsilon _{\frac{{\mu_1} }{2}}}
\end{array} \right],
\!\!\!\!\!\label{Q:even}
\end{equation}
and if it is odd,
\begin{equation}
{{\bf{Q}}_{\mu_1} } = \frac{1}{{\sqrt 2 }}\left[ \begin{array}{l}
\ {{\bf{I}}_{\frac{{{\mu_1}  - 1}}{2}}}\ \ \ \ {{\bf{0}}_{(\frac{{{\mu_1}  - 1}}{2}) \times 1}}\ \ \ j{{\bf{I}}_{\frac{{{\mu_1}  - 1}}{2}}}\\
{{\bf{0}}_{1 \times (\frac{{{\mu_1}  - 1}}{2})}}\ \ {\sqrt 2} \ \ \ \ \ \ \ \ {{\bf{0}}_{1 \times (\frac{{{\mu_1}  - 1}}{2})}}\\
\ {\Upsilon _{\frac{{{\mu_1}  - 1}}{2}}}\ \ \ {{\bf{0}}_{(\frac{{{\mu_1}  - 1}}{2}) \times 1}}\ -j{\Upsilon _{\frac{{{\mu_1}  - 1}}{2}}}
\end{array} \right].
\!\!\!\!\!\label{Q:odd}
\end{equation}

Next we will determine the number of paths by the singular value decomposition (SVD) of ${{\bf{G}}_{\rm{re}}}(t)$: ${{\bf{G}}_{\rm{re}}}(t) = {{\bf{U}}_\tau }{{\bf{S}}_\tau }{\bf{V}}_\tau ^H$, where ${{\bf{U}}_\tau } = [{{\bf{u}}_1}, \cdots ,{{\bf{u}}_{\mu_1}}]$, ${{\bf{S}}_\tau } = {\rm{diag}}\{ {s_1}, \cdots ,{s_{{M_\tau }}}\} $, and ${M_\tau } = \min (\mu_1,2\mu_2)$. 
Without loss of generality, we assume the diagonal elements of ${{\bf{S}}_\tau }$ are in non-increasing order.
To determine the number of non-negligible paths, we define a set
\begin{equation}
\begin{array}{l}
{\mathfrak{M}} = \{ {m_\tau }|\lvert {{s_{{m_\tau }}}} \rvert \ge {\gamma _1},1 \le {m_\tau } \le {M_\tau }\} \\
\ \ \ \ = \{ {\begin{array}{*{20}{c}}{m_1,}&{\cdots,}&{m_P} \end{array}} \},
\end{array}
\!\!\!\!\!\label{Set:M}
\end{equation}
where the positive threshold ${\gamma _1}$ is close to zero. The number of non-negligible paths $P = {\rm{card}}({\mathfrak{M}})$.

Hereinafter, we estimate the path delay, which contains two substeps, i.e., generating a real matrix, and the EVD.
We first define a matrix ${{\bf{U}}_s} = [{{\bf{u}}_{{m_1}}}, \cdots ,{{\bf{u}}_{{m_P}}}]$, which is generated by selecting the $P$ columns from ${{\bf{U}}_\tau }$.
A real matrix related to the path delays is defined as
\begin{equation}
{{\bf{\Psi }}_\tau } = {\left[ {{\mathop{\rm Re}\nolimits} ({\bf{Q}}_{\mu_3 }^H{{\bf{J}}_1}{\bf{Q}}_{\mu_1 }^{}){{\bf{U}}_s}} \right]^\dag }{\mathop{\rm Im}\nolimits}({\bf{Q}}_{\mu_3}^H{{\bf{J}}_1}{\bf{Q}}_{\mu_1}^{}){{\bf{U}}_s},
\!\label{RealMatrix_delay}
\end{equation}
where ${\mu_3} = KRL - RL$ and ${{\bf{J}}_1} = \left[ {{{\bf{I}}_{\mu_3}}:{{\bf{0}}_{{\mu_3} \times RL}}} \right]$. The unitary matrix ${{\bf{Q}}_{\mu_3}}$ only depends on the size ${\mu_3}$, and is expressed as Eq. (\ref{Q:even}) or Eq. (\ref{Q:odd}). We also introduce a matrix 
${{\bf{\hat Z}}_\tau}$, which contains the delay information:
\begin{equation}
{{\bf{\hat Z}}_\tau} = {\rm{diag}}\left\{\!\! {\begin{array}{*{20}{c}}{\tan \left( {\pi \Delta f{\tau _1}(t)} \right),}\!\!&\!\!{\cdots,}\!\!&\!\!{\tan \left( {\pi \Delta f{\tau _{{P}}}(t)} \right)} \end{array}} \!\!\right\}.
\!\label{Tan_Delay}
\end{equation}
The two matrices ${{\bf{\Psi }}_\tau }$ and ${{\bf{\hat Z}}_\tau}$ are similarity matrices that share the same eigenvalues \cite{91KehTSP}. Then, the path delay will be estimated by the EVD of ${{\bf{\Psi }}_\tau }$: ${{\bf{\Psi }}_\tau } = {{\bf{W}}_\tau}{{\bf{\hat Z}}_\tau}{{\bf{W}}_\tau ^{ - 1}}$, where ${{\bf{W}}_\tau}$ is an eigenvectors collection of ${{\bf{\Psi }}_\tau }$. 
Since ${{\bf{\Psi }}_\tau }$ and ${{\bf{\hat Z}}_\tau}$ are similarity matrices,  ${{\bf{\hat Z}}_\tau}$ can be calculated by ${{\bf{\hat Z}}_\tau} = {\bf{W}}_\tau ^{ - 1}{{\bf{\Psi }}_\tau }{{\bf{W}}_\tau }$. 
According to the angular-frequency-domain channel structure in Eq. (\ref{Channel_UE_angular-frequency}), the estimated diagonal elements in ${{\bf{\hat Z}}_\tau}$ may reflect the multipath delay phase differences between two neighboring subcarriers. Thus, the $p$-th path delay ${\hat \tau _p}(t)$ is estimated as
\begin{equation}
{\hat \tau _p}(t) = \frac{{{{\tan }^{ - 1}}\left( {{{\hat z}_{\tau,p}}} \right)}}{{\pi \Delta f}},
\!\label{Estimation_delay}
\end{equation}
where ${\hat z_{\tau,p}}$ denotes the $p$-th estimated diagonal element of ${{\bf{\hat Z}}_\tau}$. 
Moreover, the delay changing rate ${\hat k_{{\tau _p}}}$ is estimated with two different time samples: 
\begin{equation}
{\hat k_{{\tau _p}}} = \frac{{{{\hat \tau }_p}({t_2}) - {{\hat \tau }_p}({t_1})}}{{{t_2} - {t_1}}}.
\!\label{Estimation_kdelay}
\end{equation}

\par Up to now, the multipath delays have been estimated. Next, we estimate the multipath EODs and AODs mainly by two substeps, i.e., introducing some shuffling matrices, generating two real matrices related to the EODs and the AODs. Since ${{\bf{G}}_u}(t)$ contains the angular information of the paths, we only need to permute ${{\bf{G}}_u}(t)$ to generate two new 3-D MP matrices ${\bf{G}}_h(t)$ and ${\bf{G}}_v(t)$.
First, four shuffling matrices ${{\bf{S}}_{{\rm{left}},h}}$, ${{\bf{S}}_{{\rm{right}},h}}$, ${{\bf{S}}_{{\rm{left}},v}}$ and ${{\bf{S}}_{{\rm{right}},v}}$ are introduced as 
\begin{equation}
\begin{array}{l}
\!\!\!\!\!\!\!\!{{\bf{S}}_{{\rm{left}},h}}\! =\! [{\bf{s}}(1),\! \cdots\! ,{\bf{s}}(1\! +\! \mu_1\! -\! L),{\bf{s}}(2), \cdots ,\\
\ \ {\bf{s}}(2 \!+\! \mu_1\! -\! L),\! \cdots\! ,{\bf{s}}(L),\! \cdots\! {\bf{s}}(\mu_1){]^T},
\end{array}
\!\!\!\!\label{Shuffling_matrix_h}
\end{equation}
\begin{equation}
\begin{array}{l}
\!\!\!\!\!{{\bf{S}}_{{\rm{right}},h}}\! =\! [{\bf{s}}(1),\! \cdots\! ,{\bf{s}}(\mu_2\! -\! (N_h^{{\rm{BS}}} - L)),{\bf{s}}(2), \cdots ,\\
\ \ \ \ \ \ \ {\bf{s}}(\mu_2\! -\! (N_h^{{\rm{BS}}} \!-\! L \!-\! 1)),\! \cdots\! ,{\bf{s}}(N_h^{{\rm{BS}}} \!-\! L \!+\! 1),\! \cdots\! {\bf{s}}(\mu_2){]^T},
\end{array}
\!\!\!\!\label{Shuffling_matrix_hright}
\end{equation}
\begin{equation}
\begin{array}{l}
\!\!\!\!\!\!{{\bf{S}}_{{\rm{left}},v}} \!=\!\! [{\bf{s}}(1),\! \cdots\! ,\!{\bf{s}}(L),\! \cdots\! ,{\bf{s}}(1\!\! +\!\! \mu_1\!\! -\!\! LR), \!\cdots\! ,{\bf{s}}(L\!\! +\!\! \mu_1\!\! -\!\! LR),\! \cdots\! ,\\
\ \ \ \ \ {\bf{s}}(1\!\! +\!\! (R\!\! -\!\! 1)L),\! \cdots\! ,{\bf{s}}(L\!\! +\!\! (R\!\! -\!\! 1)L),\! \cdots\! ,\\
\ \ \ \ \ {\bf{s}}(1\!\! +\!\! (R\!\! -\!\! 1)L\!\! +\!\! (K\!\! -\!\! 1)LR),\! \cdots\! ,{\bf{s}}(\mu_1){]^T},
\end{array}
\!\!\!\!\label{Shuffling_matrix_v}
\end{equation}
\begin{equation}
\begin{array}{l}
\!\!\!\!\!\!{{\bf{S}}_{{\rm{right}},v}} \!=\!\! [{\bf{s}}(1),\! \cdots\! ,\!{\bf{s}}(N_h^{{\rm{BS}}} \!\!-\!\! L\!\! +\!\! 1),\! \cdots\! ,\\
\ \ \ \ \ {\bf{s}}((N_v^{{\rm{BS}}}\!\! -\!\! R\!\! +\!\! 1)(N_h^{{\rm{BS}}} \!\!-\!\! L \!\!+\!\! 1) \!\!-\!\! (N_h^{{\rm{BS}}} \!\!-\!\! L)), \!\cdots\!,\\ 
\ \ \ \ \ {\bf{s}}((N_v^{{\rm{BS}}} - R + 1)(N_h^{{\rm{BS}}} - L + 1)),\! \cdots\! ,\\
\ \ \ \ \ {\bf{s}}(1\!\! -\!\! (N_h^{{\rm{BS}}} \!\!-\!\! L \!\!+\!\! 1) \!\!+\!\! \mu_2),\! \cdots\! ,{\bf{s}}(\mu_2){]^T},
\end{array}
\!\!\!\!\label{Shuffling_matrix_vright}
\end{equation}
where ${\bf{s}}(m), m = 1, \cdots ,LRK$, is an $LRK \times 1$ unit vector with the $m$-th element being one. 
Permuted from ${{\bf{G}}_u}(t)$, the two matrices ${\bf{G}}_h(t)$ and ${\bf{G}}_v(t)$ are expressed as ${\bf{G}}_h(t) = {{\bf{S}}_{{\rm{left}},h}}{{\bf{G}}_u}(t){{{\bf{S}}_{{\rm{right}},h}^H}}$ and ${\bf{G}}_v(t) = {{\bf{S}}_{{\rm{left}},v}}{{\bf{G}}_u}(t){{{\bf{S}}_{{\rm{right}},v}^H}}$, respectively.
Notice that in practice, the two 3-D MP matrices ${\bf{G}}_h(t)$ and ${\bf{G}}_v(t)$, as well as the two shuffling matrices ${{\bf{S}}_{{\rm{right}},h}}$ and ${{\bf{S}}_{{\rm{right}},v}}$, are not needed. In fact, the EODs and the AODs might be estimated by only introducing ${{\bf{S}}_{{\rm{left}},h}}$ and ${{\bf{S}}_{{\rm{left}},v}}$.

Then, we will focus on the generation of the two real matrices related to the EODs and the AODs. By following a similar procedure to generate ${{\bf{\Psi }}_\tau }$ from ${{\bf{G}}_u}(t)$, i.e., from Eq. (\ref{3D-Gre}) to Eq. (\ref{RealMatrix_delay}), we may obtain two real matrices ${{\bf{\Psi }}_{\phi }}$ and ${{\bf{\Psi }}_\theta }$,  which are generated from ${\bf{G}}_h(t)$ and ${\bf{G}}_v(t)$ respectively.
The matrices ${{\bf{\Psi }}_{\phi }}$ and ${{\bf{\Psi }}_\theta }$ are defined by
\begin{equation}
\begin{array}{l}
{{\bf{\Psi }}_{\phi }} = {\left( {{\mathop{\rm Re}\nolimits} ({\bf{Q}}_{\mu_4}^H{{\bf{J}}_2}{\bf{Q}}_{\mu_1}^{}){\bf{Q}}_{\mu_1}^{}{{\bf{S}}_{{\rm{left}},h}}{\bf{Q}}_{\mu_1}^H{{\bf{U}}_s}} \right)^\dag }\\
\ \ \ \ \ \ \ \ \ {\mathop{\rm Im}\nolimits}({\bf{Q}}_{\mu_4}^H{{\bf{J}}_2}{\bf{Q}}_{\mu_1}^{}){\bf{Q}}_{\mu_1}^{}{{\bf{S}}_{{\rm{left}},h}}{\bf{Q}}_{\mu_1}^H{{\bf{U}}_s},
\end{array}
\!\!\!\!\label{RealMatrix_EOD_AOD}
\end{equation}
\begin{equation}
\begin{array}{l}
{{\bf{\Psi }}_\theta } = {\left( {{\mathop{\rm Re}\nolimits} ({\bf{Q}}_{\mu_5}^H{{\bf{J}}_3}{\bf{Q}}_{\mu_1}^{}){\bf{Q}}_{\mu_1}^{}{{\bf{S}}_{{\rm{left}},v}}{\bf{Q}}_{\mu_1}^H{{\bf{U}}_s}} \right)^\dag }\\
\ \ \ \ \ \ \ {\mathop{\rm Im}\nolimits}({\bf{Q}}_{\mu_5}^H{{\bf{J}}_3}{\bf{Q}}_{\mu_1}^{}){\bf{Q}}_{\mu_1}^{}{{\bf{S}}_{{\rm{left}},v}}{\bf{Q}}_{\mu_1}^H{{\bf{U}}_s},
\end{array}
\!\!\!\!\label{RealMatrix_EOD}
\end{equation}
\noindent where ${\mu_4}=KRL - KR$, ${\mu_5}=KRL - KL$, ${{\bf{J}}_2} = \left[ {{{\bf{I}}_{\mu_4}}:{{\bf{0}}_{{\mu_4} \times KR}}} \right]$, and ${{\bf{J}}_3} = \left[ {{{\bf{I}}_{\mu_5}}:{{\bf{0}}_{{\mu_5} \times KL}}} \right]$. The two unitary matrices ${{\bf{Q}}_{\mu_4}}$ and ${{\bf{Q}}_{\mu_5}}$ are defined in Eq. (\ref{Q:even}) or Eq. (\ref{Q:odd}) according to their sizes ${\mu_4}$ and ${\mu_5}$.

Define two real diagonal matrices ${{\bf{\hat Z}}_{\phi }}$ and ${{\bf{\hat Z}}_\theta}$, which contain the angular information
\begin{equation}
\begin{array}{l}
{{{\bf{\hat Z}}}_{\phi }} = {\rm{diag}}\{ \tan (\frac{{\pi d_h^{{\rm{tx}}}\cos ({\theta _1})\sin({\phi _1})}}{{{\lambda _0}}}), \\
\ \ \ \ \ \ \ \ \cdots ,\tan (\frac{{\pi d_h^{{\rm{tx}}}\cos ({\theta _{{P}}})\sin({\phi _{{P}}})}}{{{\lambda _0}}})\},
\end{array}
\!\!\!\!\label{Tan_EOD_AOD}
\end{equation}
\begin{equation}
\begin{array}{l}
\!\!\!\!\!\!\!\!{{\bf{\hat Z}}_\theta } = {\rm{diag}}\{ \!{\begin{array}{*{20}{c}}{\tan (\frac{{\pi d_v^{{\rm{tx}}}\sin ({\theta _1})}}{{{\lambda _0}}}),}\!&\!{\cdots,}\!&\!{\tan (\frac{{\pi d_v^{{\rm{tx}}}\sin ({\theta _{{P}}})}}{{{\lambda _0}}})} \end{array}}\!\}.
\end{array}
\!\!\!\!\label{Tan_EOD}
\end{equation}
The diagonal elements in ${{\bf{\hat Z}}_{\phi }}$ reflect the phase differences between the two neighboring antenna elements in the horizontal direction of the BS antenna panel of the paths. 
Likewise, ${{\bf{\hat Z}}_\theta }$ reflects the phase differences between the two neighboring antenna elements, in the vertical direction of the BS antenna panel for the paths. 

Finally, we will estimate the EODs and the AODs as follows: The two real matrices ${{\bf{\Psi }}_{\phi }}$ and ${{\bf{\Psi }}_\theta }$ are similar to ${{\bf{\hat Z}}_{\phi }}$ and ${{\bf{\hat Z}}_\theta}$, respectively, i.e., ${{\bf{\Psi }}_{\phi }}$ and ${{\bf{\hat Z}}_{\phi }}$ share the same eigenvalues, 
${{\bf{\Psi }}_\theta }$ and ${{\bf{\hat Z}}_\theta}$ also share the same eigenvalues. Thus, we may obtain ${{\bf{\hat Z}}_{\phi }}$ and ${{\bf{\hat Z}}_\theta}$ via the EVDs of ${{\bf{\Psi }}_{\phi }}$ and  ${{\bf{\Psi }}_\theta }$. Note that ${{\bf{\Psi }}_{\phi }}$, ${{\bf{\Psi }}_\theta }$, and ${{\bf{\Psi }}_\tau }$ share the same eigenvectors, since ${\bf{G}}_h(t)$ and ${\bf{G}}_v(t)$ are simply permuted versions of ${{\bf{G}}_u}(t)$. 
More specifically, we may denote the EVDs by ${{\bf{\Psi }}_{\phi }} = {{\bf{W}}_\tau}{{\bf{\hat Z}}_{\phi }}{{\bf{W}}_\tau ^{ - 1}}$ and 
${{\bf{\Psi }}_\theta } = {{\bf{W}}_\tau }{{\bf{\hat Z}}_\theta}{{\bf{W}}_\tau ^{ - 1}}$, where ${\bf{W}}_\tau$ contains the eigenvectors of ${{\bf{\Psi }}_\tau }$. 
As a result, ${{\bf{\hat Z}}_{\phi }}$ and ${{\bf{\hat Z}}_\theta}$ are calculated as ${{\bf{\hat Z}}_{\phi }} = {\bf{W}}_\tau ^{ - 1}{{\bf{\Psi }}_{\phi }}{{\bf{W}}_\tau }$ and ${{\bf{\hat Z}}_\theta} = {\bf{W}}_\tau ^{ - 1}{{\bf{\Psi }}_\theta }{{\bf{W}}_\tau }$.
The EOD and the AOD are estimated by
\begin{equation}
{\hat \theta _p} = {\sin ^{ - 1}}\left( {\frac{{{\tan ^{-1}}({{\hat z}_{\theta,p}}) {\lambda _0}}}{{\pi d_v^{{\rm{tx}}}}}} \right),
\!\!\!\!\label{Estimation_EOD}
\end{equation}
\begin{equation}
{\hat \phi _p} = {\sin ^{ - 1}}\left( {\frac{{{\tan ^{-1}}({{\hat z}_{\phi,p}}) {\lambda _0}}}{{\pi d_h^{{\rm{tx}}} \cos \left( {{{\sin }^{ - 1}}\left( {\frac{{{\tan ^{-1}}({{\hat z}_{\theta,p}}) {\lambda _0}}}{{\pi d_v^{{\rm{tx}}}}}} \right)} \right)}}} \right),
\!\!\!\!\label{Estimation_AOD}
\end{equation}
where ${\hat z_{\phi,p}}$ and ${\hat z_{\theta,p}}$ denote the $p$-th estimated diagonal elements of ${{\bf{\hat Z}}_{\phi }}$ and ${{\bf{\hat Z}}_\theta}$, respectively. 
So far, the EODs, AODs and delays have been estimated by our proposed 3-D angle-delay estimation scheme, which is summarized in Algorithm \ref{algorithm_1}.
Notice that ${{\bf{H}}_u}(t)$ as an input of Algorithm \ref{algorithm_1} is a historical sample obtained from the UL sounding reference signal (SRS).
With Algorithm \ref{algorithm_1}, we will estimate the channel parameters, which are used for reconstructing and predicting the future DL channel.

We now analyze the computational complexity of this algorithm in terms of time and memory. The time complexity is dominated by step 3 - step 8. The complexity of step 3 is $\mathcal{O}(2\mu _1^2{\mu _2}) + \mathcal{O}(4{\mu _1}\mu _2^2)$. The main complexity of step 4 is the SVD, i.e., $\mathcal{O}(4{\mu _1}\mu _2^2)$. Step 5 contains the matrix inversion and multiplications, and has a complexity order of $\mathcal{O}(\mu _3^2{\mu _1}) + \mathcal{O}({\mu _3}\mu _1^2) + \mathcal{O}({P^{2.37}})$. In step 6, the computation complexity is dominated by the EVD, i.e., $\mathcal{O}({P^3})$. The complexity of step 7 is $\mathcal{O}(\mu _4^2{\mu _1}) + \mathcal{O}({\mu _4}\mu _1^2) + \mathcal{O}(\mu _1^3) + \mathcal{O}({P^{2.37}})$, which mainly contains the matrix inversion and multiplication. Similar to step 6, step 8 has a complexity order of $\mathcal{O}({P^3})$. According to the conditions of pencil sizes in Eq. (\ref{NC_3D-MP}), we may obtain the global time complexity of Algorithm \ref{algorithm_1} as $\mathcal{O}(2{\mu _1^2}{\mu _2}) + \mathcal{O}(4{\mu _1}{\mu _2^2}) + \mathcal{O}({\mu _1^3})$. 
Then, the memory complexity of Algorithm \ref{algorithm_1} is dominated by step 3, i.e., the two matrices ${{\bf{Q}}_{\mu_1 }}$ and ${{\bf{Q}}_{\mu_2 }}$ in step 3 lead to the total memory complexity order of $\mathcal{S}(\mu _1^2) + \mathcal{S}(4\mu _2^2)$.
\begin{algorithm}[htb]
\renewcommand{\algorithmicrequire}{\textbf{Input:}}
\renewcommand{\algorithmicensure}{\textbf{Output:}}
\caption{Proposed 3-D angle-delay estimation scheme.}
\label{algorithm_1}
\begin{algorithmic}[1]
\REQUIRE ${{\bf{H}}_u}(t)$, $L$, $R$, $K$, ${\gamma _1}$;
\FOR{$t = \left\{ {{t_1},{t_2}} \right\}$}
\STATE Generate a 3-D MP matrix ${{\bf{G}}_u}(t)$ following the procedure between Eq. (\ref{1DMP}) and Eq. (\ref{3D-MP-Gu});
\STATE Convert ${{\bf{G}}_u}(t)$ to 3-D real matrix ${{\bf{G}}_{\rm{re}}}(t)$ according to Eq. (\ref{3D-Gre}) and Eq. (\ref{3D-Gex});
\STATE Determine the number of paths by Eq. (\ref{Set:M});
\STATE Generate the real matrix ${{\bf{\Psi }}_\tau }$ with Eq. (\ref{RealMatrix_delay});
\STATE Estimate the delay ${\hat \tau _p}(t)$ as Eq. (\ref{Estimation_delay}) by the EVD of ${{\bf{\Psi }}_\tau }$;
\STATE Define two shuffling matrices as Eq. (\ref{Shuffling_matrix_h}) and Eq. (\ref{Shuffling_matrix_v}), and generate ${{\bf{\Psi }}_{\phi }}$ and ${{\bf{\Psi }}_\theta }$ with Eq. (\ref{RealMatrix_EOD_AOD}) and Eq. (\ref{RealMatrix_EOD});
\STATE Estimate the EOD ${\hat \theta _p}$ and the AOD ${\hat \phi _p}$ by Eq. (\ref{Estimation_EOD}) and Eq. (\ref{Estimation_AOD}), respectively; 
\ENDFOR
\STATE Estimate the parameter ${\hat k_{{\tau _p}}}$ by Eq. (\ref{Estimation_kdelay});
\ENSURE ${\hat \tau _p}(t)$, ${\hat \theta _p}$, ${\hat \phi _p}$, ${\hat k_{{\tau _p}}}$, $p \in \left\{ {1, \cdots ,{P}} \right\}$
\end{algorithmic}
\end{algorithm}

\subsection{The 3-D Angle-Doppler Estimation}\label{sec:angle-Doppler estimation}
\par With the angular and delay information obtained in Sec. \ref{sec:angle-delay estimation}, we now estimate the multipath angles and Doppler in this section. Since the duration of the channel sample is $T$, the ${n_s}$-th sample time is denoted by $t = {n_s}{T}$. Similar to Sec. \ref{sec:angle-delay estimation}, we start by defining a 3-D MP matrix ${{\bm{\mathcal G}}_u}({f_1}) \in {{\mathbb{C}}^{LRQ \times (N_h^{{\rm{BS}}} - L + 1)(N_v^{{\rm{BS}}} - R + 1)({N_s} - Q + 1)}}$ as
\begin{equation}
{{\bm{\mathcal G}}_u}({f_1}) = \left[ \begin{array}{l}
{{\bm{\mathcal G}}_u}(1,{f_1}), \ \cdots ,{{\bm{\mathcal G}}_u}({N_s} \!-\! Q \!+\! 1,{f_1})\\
 \ \ \ \ \vdots \ \ \ \ \ \ \ \ \ddots \ \ \ \ \ \ \ \ \ \vdots \\
{{\bm{\mathcal G}}_u}(Q,{f_1}), \cdots ,\ \ \ \ {{\bm{\mathcal G}}_u}({N_s},{f_1})
\end{array} \right],
\!\!\!\!\label{3D-MP_angular_frequency}
\end{equation}
where ${N_s}$ is the number of samples, and $Q$ is the pencil size satisfying $P < Q < {N_s}-P+1$.
The matrix ${{\bm{\mathcal G}}_u}({f_1})$ is composed of the channels between the $u$-th UE antenna and all BS antennas at the first subcarrier and all sample times. For simplicity, we define ${{\omega}_{\mu_1}}=LRQ$ and ${{\omega}_{\mu_2}}=(N_h^{{\rm{BS}}} - L + 1)(N_v^{{\rm{BS}}} - R + 1)({N_s} - Q + 1)$.
The 2-D MP matrix ${{\bm{\mathcal G}}_u}({n_s},{f_1}) \in {{\mathbb{C}}^{LR \times (N_h^{{\rm{BS}}} - L + 1)(N_v^{{\rm{BS}}} - R + 1)}}$, ${n_s} = 1, \cdots ,{N_s}$ consists of the channels between the $u$-th UE antenna and all BS antennas at the first subcarrier and the ${n_s}$-th sample. 
We rewrite ${{\bm{\mathcal G}}_u}({n_s},{f_1})$ as
\begin{equation}
{{\bm{\mathcal G}}_u}({n_s},{f_1}) = {{\check{\bf{E}}}_{1}}{{\bf{Y}}_w}{\bf{Z}}_{w_\tau}^{{n_s}}{{\check{\bf{F}}}_{1}},
\!\!\!\!\label{Gu_f1_nt}
\end{equation}
where ${\bf{Z}}_{\omega_\tau}$ denotes the phase differences between two neighboring samples of the paths:
\begin{equation}
{\bf{Z}}_{\omega_\tau} = {\rm{diag}}\left\{ {\begin{array}{*{20}{c}}{e^{j2\pi {\omega_{\tau _1}} {T}},}&{\cdots,}&{e^{j2\pi {\omega_{\tau _P}} {T}}} \end{array}} \right\},
\!\!\!\!\label{Tan_doppler}
\end{equation}
where
\begin{equation}
{\omega_{\tau _p}} = {\omega_p} - {f_1} {k_{{\tau _p}}}.
\!\!\!\!\label{doppler_f1_kdelay}
\end{equation}
The matrix ${{\bf{Y}}_\omega}$ is defined as
\begin{equation}
{{\bf{Y}}_\omega} = {\rm{diag}}\left\{ {{{\bf{D}}_u}{\bf{e}}({f_1})} \right\},
\!\!\!\!\label{S_doppler}
\end{equation}
where ${{\bf{D}}_u}$ and ${\bf{e}}({f_1})$ are shown in Eq. (\ref{Diagonal_Du}) and Eq. (\ref{Delay_initial_vector}).
According to the angular-time-domain channel structure in Eq. (\ref{Channel_UE_angular-time}), ${{\bm{\mathcal G}}_u}({f_1})$ is rewritten as
\begin{equation}
{{\bm{\mathcal G}}_u}({f_1}) = {{\check{\bf{E}}}_{\omega,2}}{{\bf{Y}}_\omega}{{\check{\bf{F}}}_{\omega,2}},
\!\!\!\!\label{3D-MP_angular_frequency_f1}
\end{equation}
where 
${{\check{\bf{E}}}_{\omega,2}} \in {{\mathbb{C}}^{{{\omega}_{\mu_1}} \times P}}$ and ${{\check{\bf{F}}}_{\omega,2}} \in {{\mathbb{C}}^{P \times {{\omega}_{\mu_2}}}}$ are defined as
\begin{equation}
{{\check{\bf{E}}}_{\omega,2}} = {\left[\! {\begin{array}{*{20}{c}}{({{\check{\bf{E}}}_{1}}{\bf{Z}}_{\omega_\tau}^{0})^T,}\!&\!{\cdots,}\!&\!{({{\check{\bf{E}}}_{1}}{\bf{Z}}_{\omega_\tau}^{Q - 1})^T} \end{array}} \!\right]^T},
\!\!\!\!\label{Eleft2_doppler}
\end{equation}
\begin{equation}
{{\check{\bf{F}}}_{\omega,2}} = \left[ {\begin{array}{*{20}{c}}{{\bf{Z}}_{\omega_\tau}^{0}{{\check{\bf{F}}}_{1}},}&{\cdots,}&{{\bf{Z}}_{\omega_\tau}^{{N_s} - Q}{{\check{\bf{F}}}_{1}}} \end{array}} \right].
\!\!\!\!\label{Eright2_doppler}
\end{equation}
The physical meaning of ${{\check{\bf{E}}}_{1}}{\bf{Z}}_{\omega_\tau}^{q_1}$, ${q_1} = 0, \cdots ,Q - 1$ is the 3-D steering vectors of a subset of antenna array with size $R \times L$ on the BS antenna panel at the $({q_1} + 1)$-th sample. 
Likewise, ${\bf{Z}}_{\omega_\tau}^{q_2}{{\check{\bf{F}}}_{1}}$, ${q_2} = 0, \cdots ,{N_s} - Q$ is the 3-D steering vectors of a subset of antenna array with size $(N_v^{{\rm{BS}}} - R + 1) \times (N_h^{{\rm{BS}}} - L + 1)$ on the BS antenna panel at the $({q_2} + 1)$-th sample.

The three pencil sizes in ${{\bm{\mathcal G}}_u}({f_1})$ yield
\begin{equation}
\left\{ \begin{array}{l}
\ \ \ \ \ \ \ \ \ \ \ \ \ \ \ \ \ \ LR(Q - 1) \ge P\\
\ \ \ \ \ \ \ \ \ \ \ \ \ \ \ \ \ \ LQ(R - 1) \ge P\\
\ \ \ \ \ \ \ \ \ \ \ \ \ \ \ \ \ \ RQ(L - 1) \ge P\\
(N_h^{{\rm{BS}}} - L + 1)(N_v^{{\rm{BS}}} - R + 1)({N_s} - Q + 1) \ge P
\end{array} \right.\!\!\!\!\!\!.
\!\!\!\!\label{NC_3D-MP_angular_frequency}
\end{equation}
\par By following the similar estimation procedures of Sec. \ref{sec:angle-delay estimation}, we then transform ${{\bm{\mathcal G}}_u}({f_1})$ to a real matrix as ${{\bm{\mathcal G}}_{\rm{re}}}({f_1})$. The step of determining the number of paths is not needed, because this work has been finished in Sec. \ref{sec:angle-delay estimation}. Next, we will estimate the parameter ${\hat \omega_{\tau _p}}$ by generating a real matrix and the EVD. More specifically, a real matrix related to Doppler is given by
\begin{equation}
\ \!\!{{\bf{\Psi }}_{\omega_\tau}} \!\!=\!\! {\left( \!{{\mathop{\rm Re}\nolimits} ({\bf{Q}}_{{\omega}_{\mu_3}}^H{{\bf{J}}_{\omega,1}}{{\bf{Q}}_{{\omega}_{\mu_1}}}){{\bf{U}}_{\omega,s}}} \!\right)^\dag }\!{\rm Im}({\bf{Q}}_{{\omega}_{\mu_3}}^H{{\bf{J}}_{\omega,1}}{{\bf{Q}}_{{\omega}_{\mu_1}}}){{\bf{U}}_{\omega,s}},\!\!
\!\!\label{RealMatrix_doppler_f1_kdelay}
\end{equation}
where ${{\omega}_{\mu_3}} = QRL - RL$, and ${{\bf{J}}_{\omega,1}}= \left[ {{{\bf{I}}_{{\omega}_{\mu_3}}}:{{\bf{0}}_{{{\omega}_{\mu_3}} \times RL}}} \right]$. 
The two unitary matrices ${{\bf{Q}}_{{\omega}_{\mu_1}}}$ and ${{\bf{Q}}_{{\omega}_{\mu_3}}}$ depending on the size of ${{\omega}_{\mu_1}}$ and ${{\omega}_{\mu_3}}$, are expressed as Eq. (\ref{Q:even}) or Eq. (\ref{Q:odd}).
Perform the SVD of ${{\bm{\mathcal G}}_{\rm{re}}}({f_1})$: ${{\bm{\mathcal G}}_{\rm{re}}}({f_1}) = {{\bf{U}}_\omega}{{\bf{S}}_\omega}{\bf{V}}_\omega^H$. We select the ${P}$ columns from ${{\bf{U}}_\omega}$ as ${{\bf{U}}_{\omega,s}}$, which follows the similar procedure of ${{\bf{U}}_s}$ in Sec. \ref{sec:angle-delay estimation}.
The matrix ${{\bf{\Psi }}_{\omega_\tau}}$ and a real diagonal matrix ${{\bf{\hat Z}}_{\omega_\tau}}$ are similarity matrices, and share the same eigenvalues. The matrix ${{\bf{\hat Z}}_{\omega_\tau}}$ containing the Doppler information is defined as
\begin{equation}
{{{\bf{\hat Z}}}_{{\omega _\tau }}} = {\rm{diag}}\{\! {\begin{array}{*{20}{c}}{- \tan ({\omega _{{\tau _1}}}\pi T),}\!&\!{\cdots,}\!&\!{- \tan ({\omega _{{\tau _{{P}}}}}\pi T)} \end{array}} \!\}.
\!\!\!\!\label{Z_doppler_delay}
\end{equation}
The diagonal elements in ${{{\bf{\hat Z}}}_{{\omega _\tau }}}$ reflect the phase differences between two samples of the paths. 

Until now, the real matrix ${{\bf{\Psi }}_{\omega_\tau}}$ has been generated, and the following step is the EVD.
We perform the EVD of ${{\bf{\Psi }}_{\omega_\tau}}$: ${{\bf{\Psi }}_{\omega_\tau}} = {{\bf{W}}_{\omega_\tau}}{{\bf{\hat Z}}_{\omega_\tau}}{{\bf{W}}_{\omega_\tau}^{ - 1}}$, where ${{\bf{W}}_{\omega_\tau}}$ is composed of the eigenvectors of ${{\bf{\Psi }}_{\omega_\tau}}$.
The matrix ${{\bf{\hat Z}}_{\omega_\tau}}$ is calculated as ${{\bf{\hat Z}}_{\omega_\tau}} = {\bf{W}}_{\omega_\tau}^{ - 1}{{\bf{\Psi }}_{\omega_\tau}}{{\bf{W}}_{\omega_\tau}}$.
The parameter ${\hat \omega_{\tau _p}}$ is thus estimated by
\begin{equation}
{\hat \omega_{\tau _p}} =  - \frac{{{{\tan }^{ - 1}}({{\hat z}_{\omega_\tau,p}})}}{{\pi {T}}},
\!\!\!\!\label{Estimation_doppler_f1_kdelay}
\end{equation}
where ${\hat z_{\omega_\tau,p}}$ denotes the $p$-th estimated diagonal element of ${{\bf{\hat Z}}_{\omega_\tau}}$. 
Likewise, we will estimate the EODs and the AODs as follows: The two real matrices related to the EODs and the AODs are given by
\begin{equation}
\begin{array}{l}
\!\!\!\!\!\!\!\!{{\bf{\Psi }}_{\omega_\tau,\phi }} \!=\!\! {\left( \!{{\mathop{\rm Re}\nolimits} ({\bf{Q}}_{{\omega}_{\mu_4}}^H\!{{\bf{J}}_{\omega,2}}{{\bf{Q}}_{{\omega}_{\mu_1}}}\!){\bf{Q}}_{{\omega}_{\mu_1}}\!{{\bf{S}}_{{\rm{left}},\omega,h}}{\bf{Q}}_{{\omega}_{\mu_1}}^H\!\!{{\bf{U}}_{\omega,s}}} \!\right)^\dag }\\
\ \ \ \ \ \ \ {\mathop{\rm Im}\nolimits}({\bf{Q}}_{{\omega}_{\mu_4}}^H\!{{\bf{J}}_{\omega,2}}{\bf{Q}}_{{\omega}_{\mu_1}}^{}\!){\bf{Q}}_{{\omega}_{\mu_1}}^{}\!\!{{\bf{S}}_{{\rm{left}},\omega,h}}{\bf{Q}}_{{\omega}_{\mu_1}}^H{{\bf{U}}_{\omega,s}},
\end{array}
\!\!\label{RealMatrix_EOD_f1_kdelay}
\end{equation}
\begin{equation}
\begin{array}{l}
\!\!\!\!\!\!\!{{\bf{\Psi }}_{\omega_\tau,\theta }} \!=\!\! {\left( \!{{\mathop{\rm Re}\nolimits} ({\bf{Q}}_{{\omega}_{\mu_5}}^H{{\bf{J}}_{\omega,3}}{{\bf{Q}}_{{\omega}_{\mu_1}}}\!){\bf{Q}}_{{\omega}_{\mu_1}}{{\bf{S}}_{{\rm{left}},\omega,v}}{\bf{Q}}_{{\omega}_{\mu_1}}^H{{\bf{U}}_{\omega,s}}} \!\right)^\dag }\\
\ \ \ \ \ \ {\mathop{\rm Im}\nolimits}({\bf{Q}}_{{\omega}_{\mu_5}}^H{{\bf{J}}_{\omega,3}}{\bf{Q}}_{{\omega}_{\mu_1}}^{}\!){\bf{Q}}_{{\omega}_{\mu_1}}^{}{{\bf{S}}_{{\rm{left}},\omega,v}}{\bf{Q}}_{{\omega}_{\mu_1}}^H{{\bf{U}}_{\omega,s}},
\end{array}
\!\!\label{RealMatrix_AOD_f1_kdelay}
\end{equation}
where ${{\omega}_{\mu_4}}=QRL - QR$, ${{\omega}_{\mu_5}}=QRL - QL$, ${{\bf{J}}_{\omega,2}}=\left[ {{{\bf{I}}_{{\omega}_{\mu_4}}}:{{\bf{0}}_{{{\omega}_{\mu_4}} \times QR}}} \right]$, and ${{\bf{J}}_{\omega,3}} = \left[ {{{\bf{I}}_{{\omega}_{\mu_5}}}:{{\bf{0}}_{{{\omega}_{\mu_5}} \times QL}}} \right]$.
Define ${{\bf{S}}_{{\rm{left}},\omega,h}}$ and ${{\bf{S}}_{{\rm{left}},\omega,v}}$ as two new shuffling matrices that share the similar expressions with ${{\bf{S}}_{{\rm{left}},h}}$ and ${{\bf{S}}_{{\rm{left}},v}}$ in Eq. (\ref{Shuffling_matrix_h}) and Eq. (\ref{Shuffling_matrix_v}). The two shuffling matrices ${{\bf{S}}_{{\rm{left}},\omega,h}}$ and ${{\bf{S}}_{{\rm{left}},\omega,v}}$ are generated by replacing the parameter $K$ with $Q$ in Eq. (\ref{Shuffling_matrix_h}) and Eq. (\ref{Shuffling_matrix_v}).
The two matrices ${{\bf{\Psi }}_{\omega_\tau,\phi }}$ and ${{\bf{\Psi }}_{\omega_\tau,\theta }}$ are similar to two real diagonal matrices ${{\bf{\hat Z}}_{\omega_\tau,\phi}}$ and ${{\bf{\hat Z}}_{\omega_\tau,\theta}}$:
\begin{equation}
\begin{array}{l}
\!\!\!\!\!\!\!\!\!\!\!\!\!\!{{{\bf{\hat Z}}}_{{\omega _\tau },\theta }}\! =\! {\rm{diag}}\{ \tan (\frac{{\pi d_v^{{\rm{tx}}}\sin ({\theta _{\omega ,1}})}}{{{\lambda _0}}}), \cdots \!,\tan (\frac{{\pi d_v^{{\rm{tx}}}\sin ({\theta _{\omega ,{{P}}}})}}{{{\lambda _0}}})\},
\end{array}\!\!\!\!\!\!\!\!\!\!
\!\!\!\!\label{Z_doppler_EOD}
\end{equation}
\begin{equation}
\begin{array}{l}
{{{\bf{\hat Z}}}_{{\omega _\tau },\phi }} = {\rm{diag}}\{ \tan (\frac{{\pi d_h^{{\rm{tx}}}\cos ({\theta _{\omega ,1}}){\rm{sin}}({\phi _{\omega ,1}})}}{{{\lambda _0}}}),\\
 \ \ \ \ \ \ \ \ \ \ \cdots ,\tan (\frac{{\pi d_h^{{\rm{tx}}}\cos ({\theta _{\omega ,{{P}}}}){\rm{sin}}({\phi _{\omega ,{{P}}}})}}{{{\lambda _0}}})\},
\end{array}\!\!\!\!\!\!\!
\!\!\!\!\label{Z_doppler_AOD}
\end{equation}
which are given by ${{\bf{\hat Z}}_{\omega_\tau,\phi}} = {\bf{W}}_{\omega_\tau}^{ - 1}{{\bf{\Psi }}_{\omega_\tau,\phi }}{{\bf{W}}_{\omega_\tau}}$ and ${{\bf{\hat Z}}_{\omega_\tau,\theta}} = {\bf{W}}_{\omega_\tau}^{ - 1}{{\bf{\Psi }}_{\omega_\tau,\theta }}{{\bf{W}}_{\omega_\tau}}$.
The EOD and the AOD are estimated by
\begin{equation}
{\hat \theta _{\omega,p}} = {\sin ^{ - 1}}\left( {\frac{{{\tan ^{-1}}({{\hat z}_{\omega_\tau,\phi,p}}) {\lambda _0}}}{{\pi d_v^{{\rm{tx}}}}}} \right),
\!\!\!\!\label{Estimation_EOD_doppler}
\end{equation}
\begin{equation}
\ {\hat \phi _{\omega,p}} \!=\! {\sin ^{ - 1}}\!\left( {\frac{{{\tan ^{-1}}({{\hat z}_{\omega_\tau,\theta,p}}) {\lambda _0}}}{{\pi d_h^{{\rm{tx}}} \cos \left( {\sin{^{ - 1}}\left( {\frac{{{\tan ^{-1}}({{\hat z}_{\omega_\tau,\phi,p}}) {\lambda _0}}}{{\pi d_v^{{\rm{tx}}}}}} \right)} \right)}}} \right),
\!\!\!\!\label{Estimation_AOD_doppler}
\end{equation}
where ${\hat z_{\omega_\tau,\phi,p}}$ and ${\hat z_{\omega_\tau,\theta,p}}$ denote the $p$-th estimated diagonal elements of ${{\bf{\hat Z}}_{\omega_\tau,\phi}}$ and ${{\bf{\hat Z}}_{\omega_\tau,\theta}}$. 

The procedures of our proposed 3-D angle-Doppler estimation scheme are summarized in Algorithm \ref{algorithm_2}. 
Similar to the time and memory complexity analysis of Algorithm \ref{algorithm_1}, we may obtain the time complexity of Algorithm \ref{algorithm_2} as $\mathcal{O}(2{\omega _{{\mu _1}}^2}{\omega _{{\mu _2}}}) + \mathcal{O}(4{\omega _{{\mu _1}}}{\omega _{{\mu _2}}^2}) + \mathcal{O}({\omega _{{\mu _1}}^3})$, and the memory complexity as $\mathcal{S}({\omega _{{\mu _1}}^2}) + \mathcal{S}(4{\omega _{{\mu _2}}^2})$.
\begin{algorithm}[htb]
\renewcommand{\algorithmicrequire}{\textbf{Input:}}
\renewcommand{\algorithmicensure}{\textbf{Output:}}
\caption{Proposed 3-D angle-Doppler estimation scheme.}
\label{algorithm_2}
\begin{algorithmic}[1]
\REQUIRE ${{\bf{H}}_u}(t)$, $L$, $R$, $Q$, ${N_s}$, ${P}$;
\STATE Generate ${{\bm{\mathcal G}}_u}({f_1})$ according to the procedure between Eq. (\ref{3D-MP_angular_frequency}) and Eq. (\ref{3D-MP_angular_frequency_f1});
\STATE Convert ${{\bm{\mathcal G}}_u}({f_1})$ to a real matrix ${{\bm{\mathcal G}}_{\rm{re}}}({f_1})$;
\STATE Generate the matrix ${{\bf{\Psi }}_{\omega_\tau}}$ with Eq. (\ref{RealMatrix_doppler_f1_kdelay});
\STATE Estimate ${\hat \omega_{\tau _p}}$ by Eq. (\ref{Estimation_doppler_f1_kdelay}) with the EVD of ${{\bf{\Psi }}_{\omega_\tau}}$;
\STATE Generate two matrices ${{\bf{\Psi }}_{\omega_\tau,\phi }}$ and ${{\bf{\Psi }}_{\omega_\tau,\theta }}$ with Eq. (\ref{RealMatrix_EOD_f1_kdelay}) and Eq. (\ref{RealMatrix_AOD_f1_kdelay});
\STATE Estimate the EOD ${\hat \theta _{\omega,p}}$ and the AOD ${\hat \phi _{\omega,p}}$ by Eq. (\ref{Estimation_EOD_doppler}) and Eq. (\ref{Estimation_AOD_doppler});
\ENSURE ${\hat \omega_{\tau _p}}$, ${\hat \theta _{\omega,p}}$, ${\hat \phi _{\omega,p}}$, $p \in \left\{ {1, \cdots ,{P}} \right\}$
\end{algorithmic}
\end{algorithm}
\par Although we have obtained the angular-delay information of the paths in Algorithm \ref{algorithm_1}, and the angular-Doppler information of the paths in Algorithm \ref{algorithm_2}, the angles, delay, and Doppler of the paths have to be paired so as to reconstruct the channel.

\subsection{The Proposed Angle-Delay-Doppler Pairing Scheme }
 
\par We propose to pair the angles, delay, and Doppler of the paths, with the super-resolution property of the angle estimation. Let some ${P} \times 1$ vectors $\bm{\hat \theta }$, $\bm{\hat \phi }$, $\bm{\hat \tau (t)}$, $\bm{\hat k_\tau }$, $\bm{\hat \theta _\omega}$, $\bm{\hat \phi _\omega}$ and $\bm{\hat \omega_{\tau}}$ denote the corresponding estimated parameters of multipaths. A pairing matrix ${[{{\bf{s}}_{{\rm{pair}},1}}, \cdots ,{{\bf{s}}_{{\rm{pair}},{P}}}]}$ is defined to map the relationship between $[\bm{\hat \theta },\bm{\hat \phi }]$ and $[\bm{\hat \theta _\omega},\bm{\hat \phi _\omega}]$ as
\begin{equation}
[\bm{\hat \theta },\bm{\hat \phi }] = {[{{\bf{s}}_{{\rm{pair}},1}}, \cdots ,{{\bf{s}}_{{\rm{pair}},{P}}}]}[\bm{\hat \theta _\omega},\bm{\hat \phi _\omega}],
\!\!\!\!\label{Mapping_matrix}
\end{equation}
where ${{\bf{s}}_{{\rm{pair}},p}} = [{s_{p,1}}, \cdots ,{s_{p,{P}}}]^T$, $p = 1, \cdots ,{P}$ is a ${P} \times 1$ unit vector with only one element being one. Denote the $p$-th row of $[{\bm{\hat \theta _w}},{\bm{\hat \phi _w}}]$ as $[{\hat \theta _{w,p}},{\hat \phi _{w,p}}]$. By solving the minimization problem
\begin{equation}
{{\bf{s}}_{{\rm{pair}},p}} = \arg \mathop {\min }\limits_{{{\bf{s}}_{{\rm{pair}},p}}} (\lvert {{{\hat \theta }_{w,p}}{{\bf{1}}_{{P} \times {1}}} - \bm{\hat \theta }} \rvert + \lvert {{{\hat \phi }_{w,p}}{{\bf{1}}_{{P} \times {1}}} - \bm{\hat \phi }} \rvert),
\!\!\!\!\label{Pairing_matrix}
\end{equation}
\noindent we may obtain ${{\bf{s}}_{\text{pair},p}}$.
Specifically, $\lvert {{{\hat \theta }_{w,p}}{{\bf{1}}_{{P} \times {1}}} - \bm{\hat \theta } } \rvert$ is the error vector between ${\hat \theta _{w,p}}$ and $\bm{\hat \theta } $, and $\lvert {{{\hat \phi }_{w,p}}{{\bf{1}}_{{P} \times {1}}} - \bm{\hat \phi } } \rvert$ is the error vector between ${\hat \phi _{w,p}}$ and $\bm{\hat \phi } $. The row index of $[{\bm{\hat \theta _w}},{\bm{\hat \phi _w}}]$ satisfying the minimize value of $\lvert {{{\hat \theta }_{w,p}}{{\bf{1}}_{{P} \times {1}}} - \bm{\hat \theta } } \rvert + \lvert {{{\hat \phi }_{w,p}}{{\bf{1}}_{{P} \times {1}}} - \bm{\hat \phi } } \rvert$ is the row index of the only non-zero entry in ${{\bf{s}}_{\text{pair},p}}$.

\par By reordering the entries of $\bm{\hat \omega_{\tau}}$ with the pairing result, the Doppler $\bm{\hat \omega}$ is obtained as
\begin{equation}
\bm{\hat \omega} = {({[{{\bf{s}}_{{\rm{pair}},1}}, \cdots ,{{\bf{s}}_{{\rm{pair}},{P}}}]})}\bm{\hat \omega_{\tau}} + {f_1}{\bm{\hat k_\tau }}.
\!\!\!\!\label{Estimation_Doppler}
\end{equation}
Now, $\bm{\hat \tau (t)}$ and $\bm{\hat \omega}$ are paired correctly and associated with the corresponding paths. With the estimated parameters, the future channel at time $t$ can be reconstructed as ${{\bf{\hat H}}_u}(t)$. 

\begin{algorithm}[htb]
\renewcommand{\algorithmicrequire}{\textbf{Input:}}
\renewcommand{\algorithmicensure}{\textbf{Output:}}
\caption{Proposed MDMP channel prediction method.}
\label{algorithm_3}
\begin{algorithmic}[1]
\STATE Estimate $\bm{\hat \theta }$, $\bm{\hat \phi }$, $\bm{\hat \tau (t)}$ and $\bm{\hat k_\tau }$ with Algorithm \ref{algorithm_1};
\STATE Estimate $\bm{\hat \theta _\omega}$, $\bm{\hat \phi _\omega}$ and $\bm{\hat \omega_{\tau _p}}$ with Algorithm \ref{algorithm_2};
\FOR{$p = 1:{P}$}
\STATE Find the pairing matrix ${[{{\bf{s}}_{{\rm{pair}},1}}, \cdots ,{{\bf{s}}_{{\rm{pair}},{P}}}]^T}$ by solving the minimization problem of Eq. (\ref{Pairing_matrix});
\ENDFOR
\STATE Calculate the Doppler $\bm{\hat \omega}$ by Eq. (\ref{Estimation_Doppler});
\STATE Predict future channel ${{\bf{\hat H}}_u}(t)$ with the estimated parameters;
\end{algorithmic}
\end{algorithm}

Our proposed MDMP method is summarized in Algorithm \ref{algorithm_3}. 
The step 3 - step 6 of Algorithm \ref{algorithm_3} has a time complexity order of $\mathcal{O}({P^2})$, and a memory complexity order of $\mathcal{S}({P^2})$. The time complexity of channel reconstruction in step 7 is $\mathcal{O}({N_h^{{\rm{BS}}}}{N_v^{{\rm{BS}}}}{N_f}{P})$. The memory complexity of step 7 is $\mathcal{S}({N_h^{{\rm{BS}}}}{N_v^{{\rm{BS}}}}{N_f})$.
The time and memory complexity of the MDMP method are primarily determined by Algorithm \ref{algorithm_1} and Algorithm \ref{algorithm_2}. 
Therefore, the global time complexity of the MDMP method is $\mathcal{O}(2{\mu _1^2}{\mu _2}) + \mathcal{O}(4{\mu _1}{\mu _2^2}) + \mathcal{O}({\mu _1^3}) + \mathcal{O}(2{\omega _{{\mu _1}}^2}{\omega _{{\mu _2}}}) + \mathcal{O}(4{\omega _{{\mu _1}}}{\omega _{{\mu _2}}^2}) + \mathcal{O}({\omega _{{\mu _1}}^3})$. The MDMP method has a global memory complexity order of $\max(\mathcal{S}(\mu _1^2) + \mathcal{S}(4\mu _2^2), \mathcal{S}({\omega _{{\mu _1}}^2}) + \mathcal{S}(4{\omega _{{\mu _2}}^2}))$. 
The PAD method in \cite{Yin20JSAC} has a complexity order of $\mathcal{O}(N{N_t}{N_f}\log({N_t}{N_f})) + \mathcal{O}({N_p}{N^{2.37}}) + \mathcal{O}({N_d}{N_p}N)$, where ${N_p} \ll {N_t}{N_f}$, $N$ and ${N_d}$ denote the prediction order of the PAD method and the number of the predicted samples. 
The MDMP method may have a larger time complexity compared to the PAD method.
However, the MDMP method has better performance and is applicable in more general settings, i.e., the CSI delay does not have to be an interval multiple of the pilot interval, the historical channel samples are not necessarily neighboring ones.
The detailed proofs will be shown in the next section.


\section{Performance Analysis of the MDMP Prediction Method}\label{sec:performance analysis}
\par In this section, we show the performance analysis of our proposed MDMP prediction method.
Denote the observation sample at time $t$ by ${{{\bf{\tilde H}}}_u}(t)$. The vectorized forms of ${{\bf{H}}_u}(t)$, ${{{\bf{\tilde H}}}_u}(t)$ and ${{\bf{\hat H}}_u}(t)$ are denoted respectively by ${{\bf{h}}_u}(t)$, ${{{{\bf{\tilde h}}}_u}(t)}$, and ${{{{\bf{\hat h}}}_u}(t)}$, i.e., ${{\bf{h}}_u}(t)\! =\! {\rm{vec}}({{\bf{H}}_u}(t))$, ${{{\bf{\tilde h}}}_u}(t) = {\rm{vec(}}{{{\bf{\tilde H}}}_u}(t))$ and ${{\bf{\hat h}}_u}(t) = {\rm{vec(}}{{{\bf{\hat H}}}_u}(t))$. 
Before the asymptotic performance analysis, we introduce an assumption.
\begin{assumption}\label{assumption1}
The observation sample yields
\begin{equation}
{{\bf{\tilde h}}_u}(t) = {{\bf{h}}_u}(t) + {\bf{n}},
\!\!\!\!\label{NoiseGaussian}
\end{equation}
where ${\bf{n}} = {[{n_1}, \cdots ,{n_{N_h^{{\rm{BS}}}N_v^{{\rm{BS}}}{N_f}}}]^T} \in {{\mathbb{C}}^{N_h^{{\rm{BS}}}N_v^{{\rm{BS}}}{N_f} \times 1}}$ is the independent identically distributed (i.i.d.) Gaussian white noise with zero mean and element-wise variance ${\sigma ^2}$. As $N_h^{{\rm{BS}}},N_v^{{\rm{BS}}},{N_f} \to \infty $, the variance ${\sigma ^2}$ converges to zero, such that:
\begin{equation}
\mathop {\lim }\limits_{N_v^{{\rm{BS}}},N_h^{{\rm{BS}}},{N_f} \to \infty } \frac{{\left\| {{{{\bf{\tilde h}}}_u}(t) - {{\bf{h}}_u}(t)} \right\|_2^2}}{{\left\| {{{\bf{h}}_u}(t)} \right\|_2^2}} = 0.
\!\!\label{Samples_accurate}
\end{equation}
\end{assumption}
Remarks: This technical assumption means the normalized channel sample error converges to zero when the number of the BS antennas and the bandwidth increase. In fact, the condition of Eq. (\ref{Samples_accurate}) can be achieved even in the multi-user multi-cell scenario with pilot contamination, by some non-linear signal processing technologies \cite{Yin16TSP}. 

For simplicity, we also introduce a vector in Definition \ref{definition1}. 
\begin{definition}\label{definition1}
Define a vector containing the parameters of multipaths as
\begin{equation}
{\bf{\Omega}}(t) = {\left[ {{{\bf{\Omega }}_1}(t), \cdots ,{{\bf{\Omega }}_P}(t)} \right]^T}, 
\!\!\!\!\label{Omega_allparmeters}
\end{equation}
where ${{\bf{\Omega }}_p}(t) = \left[ {\cos {\theta _p}\sin{\phi _p},\sin{\theta _p},{\omega _p},{\tau _p}(t)} \right]$, $p = 1, \cdots, P$.
\end{definition}

\begin{theorem}\label{theorem1}
Under Assumption \ref{assumption1}, for an arbitrary CSI delay ${t_\tau }$, the asymptotic performance of MDMP prediction method yields: 
\begin{equation}
\mathop {\lim }\limits_{N_h^{{\rm{BS}}},N_v^{{\rm{BS}}},{N_f} \to \infty }\!\!\!\!\!\! \frac{{\left\| {{{{{\bf{\hat h}}}_u}(t + {t_\tau })} \!-\! {{{\bf{h}}_u}(t + {t_\tau })}} \right\|_2^2}}{{\left\| {{{{\bf{h}}_u}(t + {t_\tau })}} \right\|_2^2}} = 0,
\!\!
\end{equation}
providing that the pencil sizes satisfy $N_h^{{\rm{BS}}} - P + 1 > L > P$, $N_v^{{\rm{BS}}} - P + 1 > R > P$, ${N_f} - P + 1 > K > P$, and ${N_s} - P + 1 > Q > P$.
\end{theorem}
\begin{proof}
The proof can be found in Appendix A. 
\end{proof}
Remarks: Theorem \ref{theorem1} indicates that the channel prediction error converges to zero when the number of the BS antennas and the bandwidth are large. Note that the CSI delay in Theorem \ref{theorem1} does not have to be an integral multiple of the pilot interval, which indicates our proposed MDMP method is more general than the PAD method in \cite{Yin20JSAC} because the PAD achieves asymptotically error-free  performance when the CSI delay is an integral multiple of the channel sampling interval.

\par If $Q = P+1$, the least number of samples in Theorem 1 should satisfy ${N_s} = 2P+1$.
The number of paths $P$ might be large in rich scattering environment. In such cases, Theorem 1 indicates that we may need a large number of samples $N_s$ to achieve asymptotically error-free prediction. In the following, we will show the asymptotic performance of the prediction error with only two arbitrary samples known.

\begin{theorem}\label{theorem22}  
Under Assumption \ref{assumption1}, if the EOD and the AOD satisfy $({\theta _p},{\phi _p}) \ne ({\theta _q},{\phi _q})$, $\forall p \ne q$, the number of subcarriers satisfies ${N_f} \ge 2$, and the configuration of the BS antennas satisfies $N_h^{{\rm{BS}}} - P + 1 > L > P$ and $N_v^{{\rm{BS}}} - P + 1 > R > P$, then with only two arbitrary samples ${{{\bf{\tilde h}}}_u}({n_{{s_1}}}{T})$ and ${{{\bf{\tilde h}}}_u}({n_{{s_2}}}{T})$ known, the asymptotic performance of the MDMP prediction method yields:
\begin{equation}
\mathop {\lim }\limits_{N_h^{{\rm{BS}}},N_v^{{\rm{BS}}} \to \infty } \frac{{\left\| {{{{{\bf{\hat h}}}_u}(t+ {t_\tau })} \!-\! {{{\bf{h}}_u}(t+ {t_\tau })}} \right\|_2^2}}{{\left\| {{{{\bf{h}}_u}(t+ {t_\tau })}} \right\|_2^2}} = 0.
\!\!\!\!\label{theorem2}
\end{equation}
\end{theorem}
\begin{proof}
The proof can be found in Appendix B. 
\end{proof}
Remarks: Theorem \ref{theorem22} requires that the angles of any two paths are different, which is a stronger assumption than in Theorem \ref{theorem1}. However a better result is obtained since only two samples are needed. 
Note that the two samples are not necessarily neighboring ones, which is also more general than the PAD method.

In Theorem \ref{theorem1} and Theorem \ref{theorem22}, the conditions of the BS antennas, e.g., $N_h^{{\rm{BS}}} - P + 1 > L > P$ and $N_v^{{\rm{BS}}} - P + 1 > R > P$, ensure the rank of the 3-D MP matrices satisfy $r({{{\bf{G}}}_{{\rm{re}}}}(t + {t_\tau })) = r({{\bm{\mathcal G}}_u}({f_1})) = P$. 
If $L = P+1$ and $R = P+1$, the least number of the BS antennas is ${N_t} = (2P+1)^2$, which seems to be large.
Next, by assuming $r({{{\bf{G}}}_{{\rm{re}}}}(t + {t_\tau })) = r({{\bm{\mathcal G}}_u}({f_1})) = P$ known, we derive a lower-bound of the BS antennas to give a satisfactory performance. In order to do so, 
we introduce two functions.
\begin{definition}\label{definition2}  
Define two functions ${F_1}(L,R)$ and ${F_2}(R)$.
\begin{equation}
{F_1}(L,R) = \frac{{\cal {N}}_s}{{N_v^{{\rm{BS}}} - R + 1}} + L - 1,
\!\!\!\!\label{definition_F1}
\end{equation}
where $L$ and $R$ satisfy
\begin{equation}
\left\{ \begin{array}{l}
LR \ge {\cal Q}\\
L \ge 2\\
N_v^{{\rm{BS}}} \ge R \ge 2
\end{array} \right.,
\!\!\!\!\label{Condition L R F1}
\end{equation}
${\cal {N}}_s = \max (\frac{Q}{{{N_s} - Q + 1}},1)$ and ${\cal Q} = \max (\frac{P}{{Q - 1}},4)$. The other function ${F_2}(R)$ is defined as
\begin{equation}
{F_2}(R) \!=\! \frac{P}{{({N_s} \!-\! Q \!+\! 1)(N_v^{{\rm{BS}}} \!-\! R \!+\! 1)}} \!+\! \frac{P}{{(R \!-\! 1)Q}} \!-\! 1,
\!\!\!\!\label{definition_F2}
\end{equation}
with
\begin{equation}
\left\{ \begin{array}{l}
LQ \ge {{\max \left( {\frac{P}{{R - 1}},4} \right)}}\\
N_v^{{\rm{BS}}} \ge R \ge 2
\end{array} \right..
\!\!\!\!\label{Condition L R F2}
\end{equation}
\end{definition}
\par Next, we derive a lower-bound of the number of the BS antennas in Proposition \ref{proposition1} to achieve a satisfactory performance.
For an arbitrary given $N_v^{{\rm{BS}}}$, we derive the bound of the total number of the BS antennas $N_t$.
In a wideband system, ${N_f}$ and ${N_s}$ usually satisfy ${N_f} > {N_s}$. According to Eq. (\ref{NC_3D-MP}) and Eq. (\ref{NC_3D-MP_angular_frequency}), the lower-bound of the number of the BS antennas is determined by ${N_s}$ in Eq. (\ref{NC_3D-MP_angular_frequency}).
The lower-bound of $N_t$ consists of two sub-bounds ${f_{N_v^{{\rm{BS}}},{N_s},1}}$ and ${f_{N_v^{{\rm{BS}}},{N_s},2}}$, 
where 
${f_{N_v^{{\rm{BS}}},{N_s},1}}$ is derived from the 3-D MP matrix in estimating Doppler, and ${f_{N_v^{{\rm{BS}}},{N_s},2}}$ is derived from the 3-D MP matrix in estimating angles.

\begin{Proposition}\label{proposition1} 
In a wideband channel, for a given ${N_v^{{\rm{BS}}}}$, a lower-bound of the number of the BS antennas $N_t$ is given as
\begin{equation}
N_t \ge \max ({f_{N_v^{{\rm{BS}}},{N_s},1}},{f_{N_v^{{\rm{BS}}},{N_s},2}}),
\!\!\!\!\label{Proposition1}
\end{equation}
where  
\begin{equation}
\begin{array}{l}
\!\!\!\!\!{f_{N_v^{{\rm{BS}}},{N_s},1}}=\!\! \max (\!{N_v^{{\rm{BS}}}}{F_1}({L_1},\!{R_1}),\!{N_v^{{\rm{BS}}}}{F_1}({L_2},\!{R_2}),\\
\!{N_v^{{\rm{BS}}}}{F_1}({L_3},\!{R_3}),{N_v^{{\rm{BS}}}}{F_1}({L_4},\!{R_4}),\!{N_v^{{\rm{BS}}}}{F_1}({L_5},\!{R_5})),\!\!\!
\end{array}
\!\!\!\!\label{Proposition1_f1}
\end{equation}
and $({L_{{n_1}}},{R_{{n_1}}})$, ${n_1} = 1,2,3,4,5$ are defined as
\begin{equation}
\begin{array}{l}
{R_1} = \frac{{\sqrt {{\cal Q}} }}{{\sqrt {{\cal {N}}_s}  + \sqrt {{\cal Q}} }}(N_v^{{\rm{BS}}} + 1)\\
{L_1} = \frac{{\sqrt {{\cal Q}} (\sqrt {{\cal {N}}_s}  + \sqrt {{\cal Q}} )}}{{N_v^{{\rm{BS}}} + 1}}
\end{array},
\!\!\!\!\label{R1-L1}
\end{equation}
\begin{equation}
\begin{array}{l}
{R_2} = \frac{{\sqrt {{\cal Q}} }}{{\sqrt {{\cal Q}}  - \sqrt {{\cal {N}}_s} }}(N_v^{{\rm{BS}}} + 1)\\
{L_2} = \frac{{\sqrt {{\cal Q}} (\sqrt {{\cal Q}}  - \sqrt {{\cal {N}}_s} )}}{{N_v^{{\rm{BS}}} + 1}}
\end{array},
\!\!\!\!\label{R2-L2}
\end{equation}
and $({R_3} = \frac{{{\cal Q}}}{2}, {L_3} = 2)$, $({R_4} = 2, {L_4} = \frac{{{\cal Q}}}{2})$, and $({R_5} = N_v^{{\rm{BS}}}, {L_5} = \frac{{{\cal Q}}}{N_v^{{\rm{BS}}}})$. The value
${f_{N_v^{{\rm{BS}}},{N_s},2}}$ is defined as
\begin{equation}
\begin{array}{l}
\!\!{f_{N_v^{{\rm{BS}}},{N_s},2}} \!\!=\\
\!\!\left\{ \begin{array}{l}
\!\!\!\!\!{N_v^{{\rm{BS}}}}(\frac{1}{{({N_s} - 1)}} + 1),\ \ \!\!\!\!\!{\rm{if}}\ \!N_v^{{\rm{BS}}} \!-\! P \!+\! 1 \!\ge\! R \!\ge\! \frac{P}{4} \!+\! 1\\
\!\!\!\!\!{N_v^{{\rm{BS}}}}(\frac{1}{{({N_s} - Q + 1)}} \!+\! \frac{P}{Q} \!-\! 1),
\!{\rm{if}}\ \!\min (N_v^{{\rm{BS}}} \!-\!\! P\! \!+\! 1,\!\frac{P}{4} \!\!+\! 1) \!\ge\! R\! \ge\! 2\\
\!\!\!\!\!{N_v^{{\rm{BS}}}}(\frac{P}{{({N_s} - 1)}} + 1),\!{\rm{if}}\ \!N_v^{{\rm{BS}}} \!\!\ge\!\! R \!\!\ge\!\! \max (N_v^{{\rm{BS}}} \!\!-\! P \!\!+\! 1,\frac{P}{4} \!+\!\! 1,2)\\
\!\!\!\!\!{N_v^{{\rm{BS}}}}{F_{2,\rm{max}}}(R),\!{\rm{if}}\ \!P \!>\! R \!\ge\! \max (N_v^{{\rm{BS}}} \!-\! P \!+\! 1,2)
\end{array} \right.\!\!\!\!\!\!,
\end{array}
\!\!\!\!\!\label{Proposition1_f2}
\end{equation}
where
\begin{equation}
\begin{array}{l}
\!\!{F_{2,\rm{max}}}(R) \!\!=\\
\!\!\left\{ \begin{array}{l}
\!\!\!\!\max (\!{F_2}(\max (2,\!{R_6})),\!{F_2}(\max (2,\!{R_7})),\\
\!{F_2}(N_v^{{\rm{BS}}})),\; \!{\rm{if}} \; Q < \frac{{{N_s} + 1}}{2}\\
\!\!\!\!\max ({F_2}(\max (2,{R_6})),{F_2}(\max (2,{R_7}))),{\rm{if}}\; Q > \frac{{{N_s} + 1}}{2}\\
\!\!\!\!{F_2}({\frac{(N_v^{{\rm{BS}}} + 1)}{2}}),\; {\rm{if}}\ \!Q = \frac{{{N_s} + 1}}{2}
\end{array} \right.\!,
\end{array}
\!\!\!\!\!\label{F2max}
\end{equation}
and
\begin{equation}
{R_6} = \frac{{\sqrt {{N_s} - Q + 1} (N_v^{{\rm{BS}}} + 1)}}{{\sqrt {{N_s} - Q + 1}  + \sqrt Q }},
\!\!\!\!\label{R6}
\end{equation}

\begin{equation}
{R_7} = \frac{{\sqrt {{N_s} - Q + 1} (N_v^{{\rm{BS}}} + 1)}}{{\sqrt {{N_s} - Q + 1}  - \sqrt Q }}.
\!\!\!\!\label{R7}
\end{equation}
\end{Proposition}
\begin{proof} 
The proof can be found in Appendix C.
\end{proof}
Remarks: 
In Proposition \ref{proposition1}, for a given $N_v^{{\rm{BS}}}$, we derive the lower-bound of $N_t$.
Such a bound means that the MDMP method does not work if the BS antenna configuration cannot satisfy the lower-bound.
As a remedy, the lower-bound can be made smaller to satisfy the antenna configuration, by increasing the number of channel samples $N_s$. 
Note that if $N_h^{{\rm{BS}}}$ is given, the lower-bound of $N_t$ can also be derived. The detailed derivation is omitted.
\section{Numerical Results}\label{sec:simulation}

\par In this section, the simulation results of our proposed method are shown. 
The clustered delay line (CDL) channel model of 3GPP is adopted. The number of scattering clusters is 9. Each cluster contains 20 rays and the total number of propagation paths is 180. According to \cite{3GPPR16}, the $p$-th path delay is modeled as ${\tau _p}(t) = \tau {'_{p,0}} - \min \left\{ {\tau {'_{1,0}}, \cdots ,\tau {'_{P,0}}} \right\} + \tau '{'_p}(t)$, where $\tau '{'_p}(t)$ is the cluster delay, and $\tau {'_{p,0}} - \min \left\{ {\tau {'_{1,0}}, \cdots ,\tau {'_{P,0}}} \right\}$ denotes the initial delay and is modeled as a spatially random variable related to the correlation distance. The root mean square (RMS) angular spreads of AOD, EOD, AOA and EOA are $87.1^\circ $, $56.4^\circ $, $102.1^\circ $ and $65.3^\circ $, respectively. The main simulation parameters are listed in Table I. 
Considering the 3-D Urban Macro (3-D UMa) scenario, the UEs have certain velocity, e.g., 60 km/h and 120 km/h. The central frequency is 3.5 GHz and the bandwidth is 100 MHz, which is composed of 273 resource blocks (RBs) in frequency domain. The duration of a slot is 0.5 ms, which contains 14 OFDM symbols. One channel sample is available for each slot. 
\begin{table}[!bht]
\caption{The main simulation parameters.}
\label{table I}
\centering
\begin{tabular}{|l|p{12em}|}
\hline
Scenario & 3D Urban Macro (3D UMa)\\
\hline
Carrier frequency (GHz) & 3.5\\
\hline
Bandwidth (MHz) & 100 (273 RBs)\\
\hline
Subcarrier spacing (kHz) & 30\\
\hline
Number of UEs & 16\\
\hline
BS antenna configuration & $(\underline{M},\underline{N}, \underline{P}, \underline{M}_g, \underline{N}_g) = (8,8,1,1,1)$, $(d_h^{{\rm{tx}}},d_v^{{\rm{tx}}})=(0.5\lambda,0.8\lambda)$\\
\hline
UE antenna configuration & $(\underline{M},\underline{N}, \underline{P}, \underline{M}_g, \underline{N}_g) = (1,1,2,1,1)$, the polarization angles are $0^\circ $ and $90^\circ $\\
\hline
Delay Spread (ns) & 616\\
\hline
CSI delay (ms) & 12, 16\\
\hline
UEs speed (km/h) & 60, 120\\
\hline
\end{tabular}
\end{table}
The antenna configuration is denoted by $(\underline{M},\underline{N}, \underline{P}, \underline{M}_g, \underline{N}_g)$, where each antenna panel has $\underline{M}$ rows and $\underline{N}$ columns of antenna elements; $\underline{P}$ is the number of polarizations; $\underline{M}_g$ is the number of panels in a row and $\underline{N}_g$ is the number of panels in a column. 
In the horizontal and vertical directions, the antenna spacings are $0.5\lambda $ and $0.8\lambda $ respectively. 
The pencil sizes such as $L$, $R$, $K$ and $Q$, are set as 6, 5, 137 and 15, respectively.
The DL precoder is eigen-based zero-forcing (EZF) \cite{Sun10TSP}.  
Two performance metrics are evaluated, which are the DL spectral efficiency and the prediction error. The DL spectral efficiency is calculated by $\sum\limits_{u = 1}^{N_\text{UE}} {E\left\{ {{{\log }_2}(1 + {\rm{SINR}}_u)} \right\}} $, where $N_\text{UE}$ is the number of UEs,  ${\rm{SINR}}_u$ is the signal-to-noise ratio of the $u$-th UE, and the expectation is taken over frequency and time. The DL prediction error is defined by $10\log \left\{ {E\left\{ {\frac{{\left\| {{{{\bf{\hat H}}}_u} - {{\bf{H}}_u}} \right\|_F^2}}{{\left\| {{{\bf{H}}_u}} \right\|_F^2}}} \right\}} \right\}$, where ${{{\bf{\hat H}}}_u}$ and ${{\bf{H}}_u}$ are the predicted and exact channels, respectively, and the expectation is taken over time, frequency and UEs.


\begin{figure}[ht]
\centering
\subfigure[]{
\includegraphics[width=3.5in]{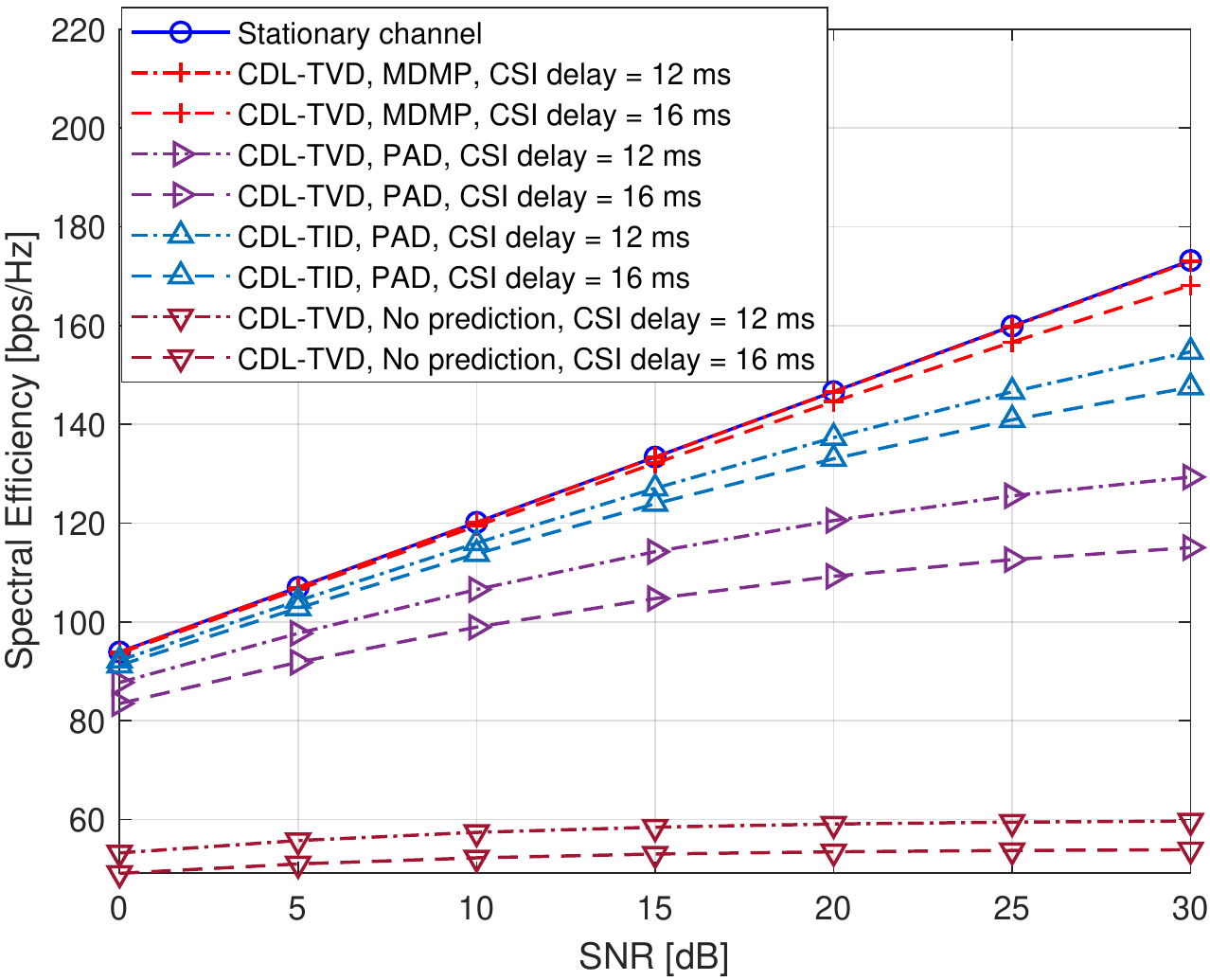}}
\\
\subfigure[]{
\includegraphics[width=3.5in]{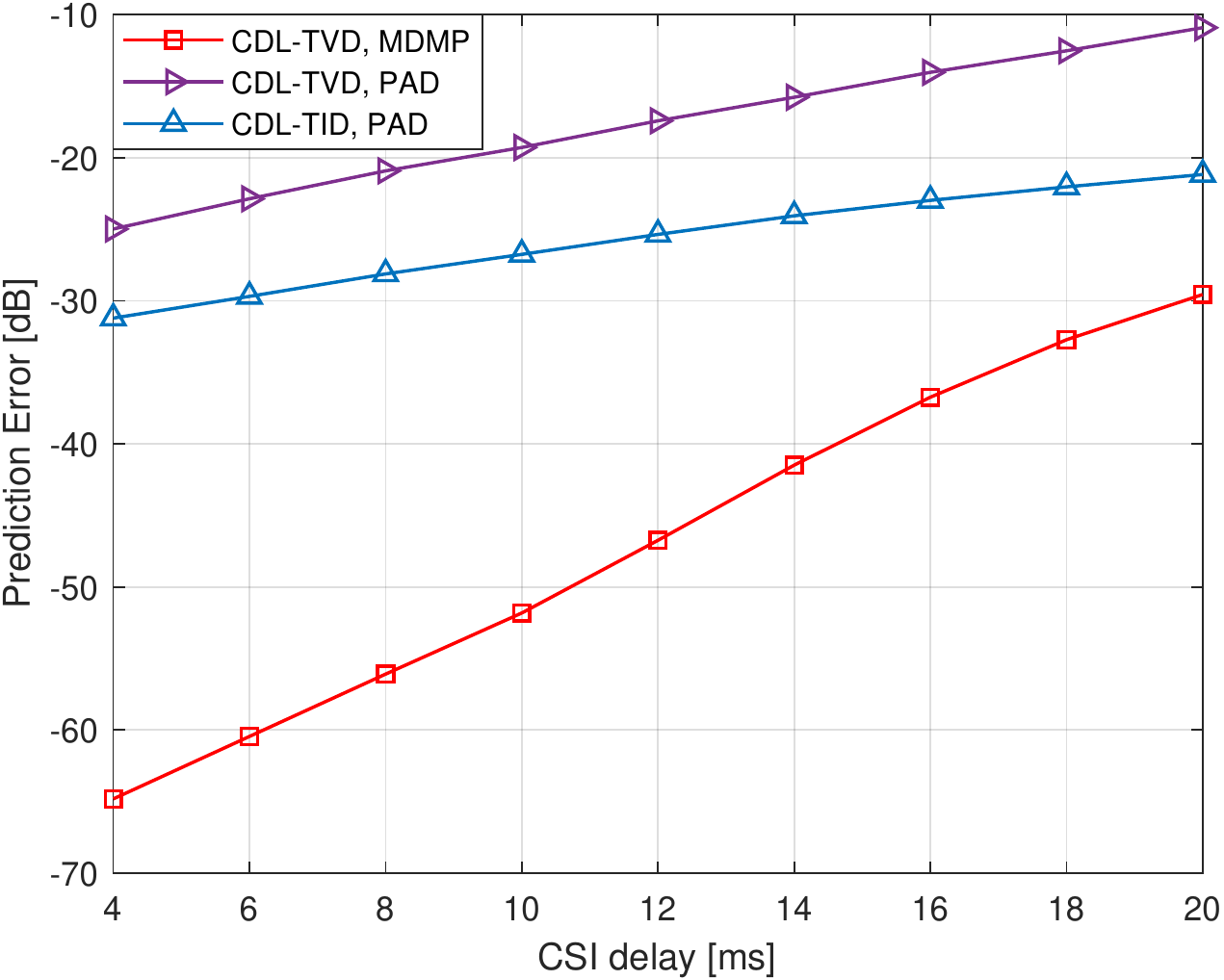}
}
\caption{(a) The spectral efficiency versus SNR and (b) the prediction error versus the CSI delay, the UEs move at 120 km/h, the BS has 64 antennas.}
\label{fig:2}
\end{figure}
\par Let CDL-TID and CDL-TVD denote the simulated CDL model with time-invariant and time-varying path delay, respectively. 
In Fig. \ref{fig:2}(a) and (b), we show the spectral efficiencies and prediction errors of different methods with the UEs moving at the speed of 120 km/h. In Fig. \ref{fig:2}(a), the CSI delay is relatively large, e.g., 12 ms and 16 ms, while in Fig. \ref{fig:2}(b), the CSI delay is within the range of 4 to 20 ms. 
The ideal case of the stationary setting is also shown as a reference curve labeled as ``Stationary channel", which is the upper-bound of the performance. The curves labeled with ``No prediction" mean channel predictions are not carried out. 
In Fig. \ref{fig:2}(a), the different performances of the PAD method in CDL-TID and CDL-TVD show that time-varying delay brings an observable decrease of the spectral efficiency. However, our proposed MDMP method is able to deal with the time-varying path delay, as it approaches the ideal case of the stationary setting in the scenario with a large CSI delay of 16 ms and high velocity level of 120 km/h.

\begin{figure}[!htb]
\centering
\includegraphics[width=3.5in]{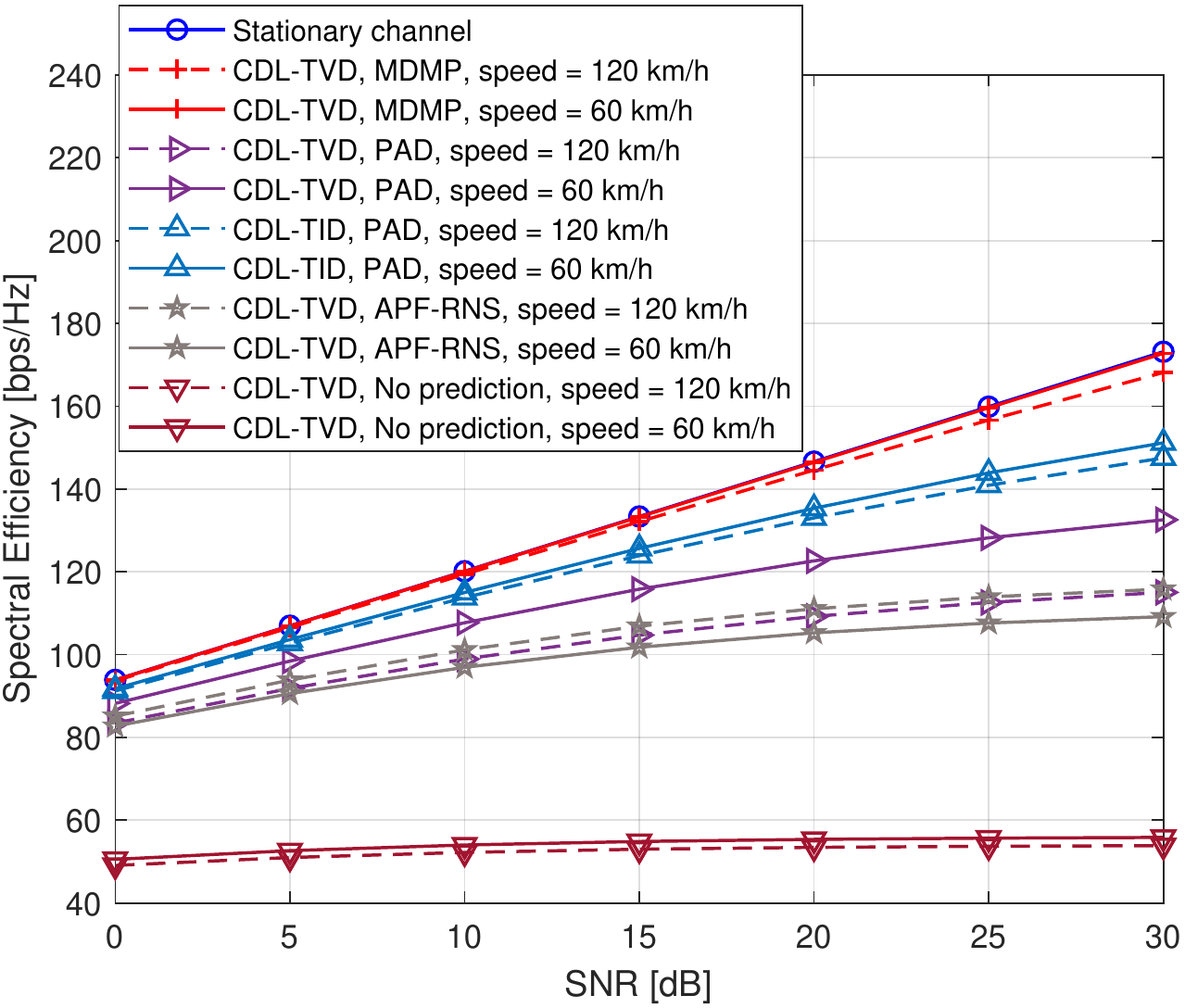}
\caption{The spectral efficiency versus SNR, 16 ms of CSI delay, the BS has 64 antennas.}
\label{fig:3}
\end{figure}
\begin{figure}[!htb]
\centering
\includegraphics[width=3.5in]{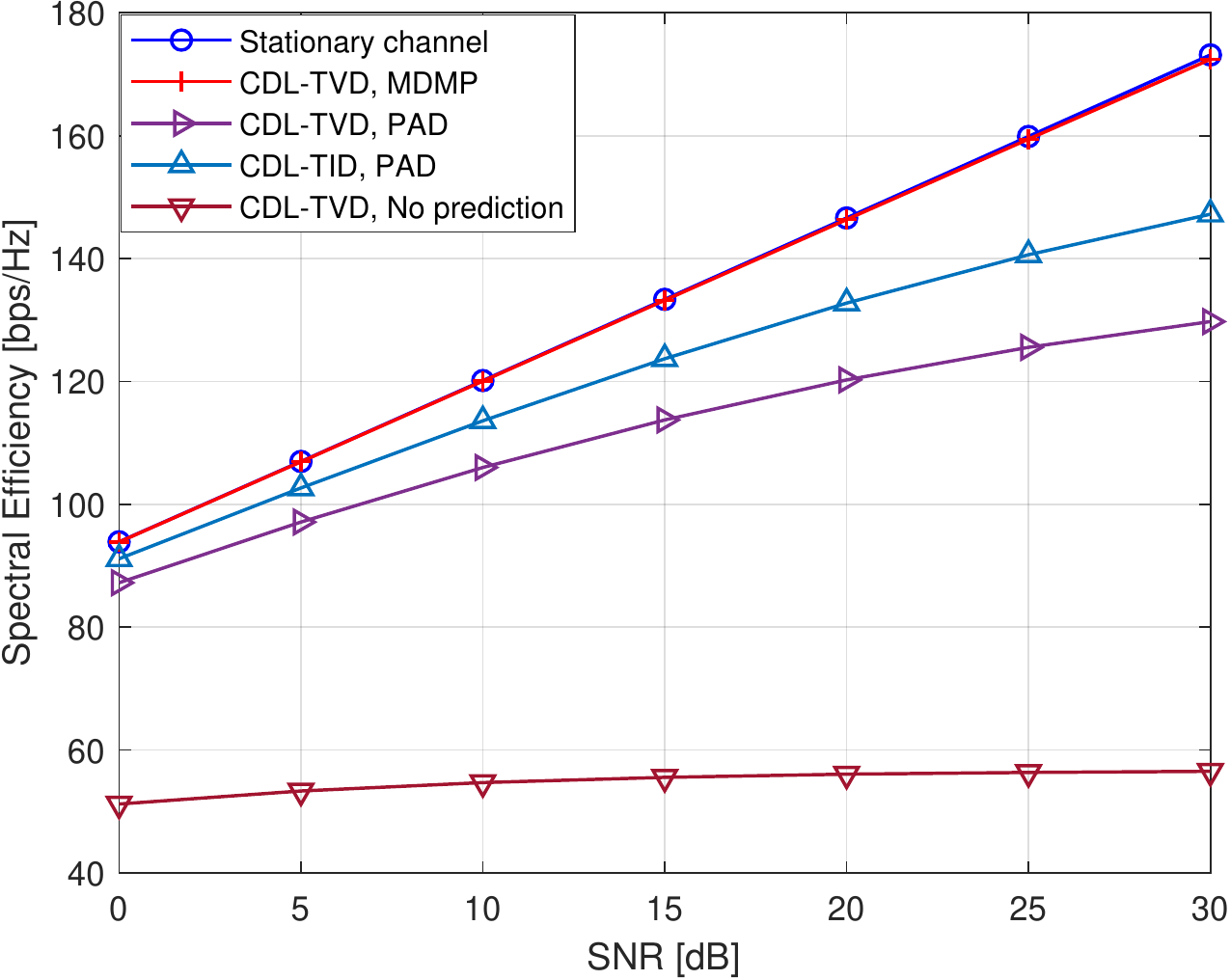}
\caption{The spectral efficiency versus SNR, 16 ms of CSI delay, the BS equipped with 64 antennas, multiple  velocity levels of UEs, i.e., four at 30 km/h, four at 60 km/h, four at 90 km/h and four at 120 km/h.}
\label{fig:4}
\end{figure}
\par Fig. \ref{fig:3} compares the spectral efficiencies of different methods with UEs moving at 60 km/h and 120 km/h, respectively. 
The curves labeled as ``APF-RNS", are the prediction performance of an adaptive and parameter-free recurrent neural structure (APF-RNS) \cite{Zhu19TCom}, which uses 19 historical sequential channel samples to predict the channel at the next moment.
In CDL-TVD, it can be observed that the performance of PAD method will decrease more obviously at higher velocity because of the faster varying delay and Doppler. However, our proposed MDMP method is close to the upper-bound of the stationary setting even in high-mobility scenarios.

\par Fig. \ref{fig:4} shows the spectral efficiencies of the MDMP and PAD methods with multiple-speed UEs, where every four UEs move at 30 km/h, 60 km/h, 90 km/h and 120 km/h, respectively. One may observe that our proposed MDMP method still performs well in this setting.

\begin{figure}[!htb]
\centering
\includegraphics[width=3.5in]{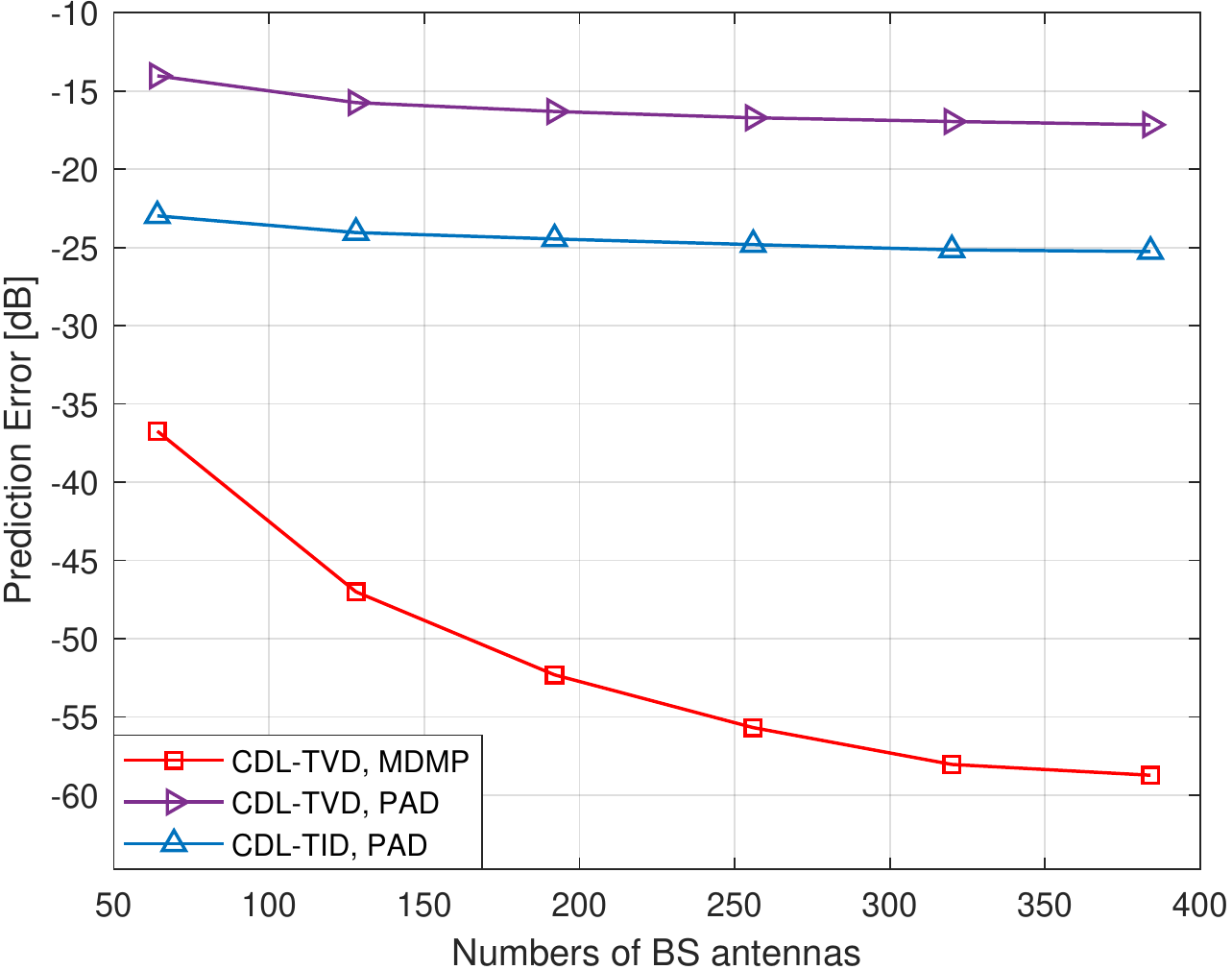}
\caption{The prediction error versus the numbers of BS antennas, UEs move at 120 km/h, 16 ms of CSI delay.}
\label{fig:5}
\end{figure}

\par Fig. \ref{fig:5} gives the prediction errors of MDMP and PAD methods with different numbers of the BS antennas, e.g., 64, 128, 192, 256, 320 and 384. The corresponding antenna configurations are $(8,8,1,1,1)$, $(8,8,1,1,2)$, $(8,8,1,1,3)$, $(8,8,1,1,4)$, $(8,8,1,1,5)$ and $(8,8,1,1,6)$, respectively. 
One may notice that the prediction error of our MDMP method decreases obviously with the increasing numbers of the BS antennas.
One may notice the error floor of the MDMP method. This is because the bandwith is fixed and not large enough.

\begin{figure}[!htb]
\centering
\includegraphics[width=3.5in]{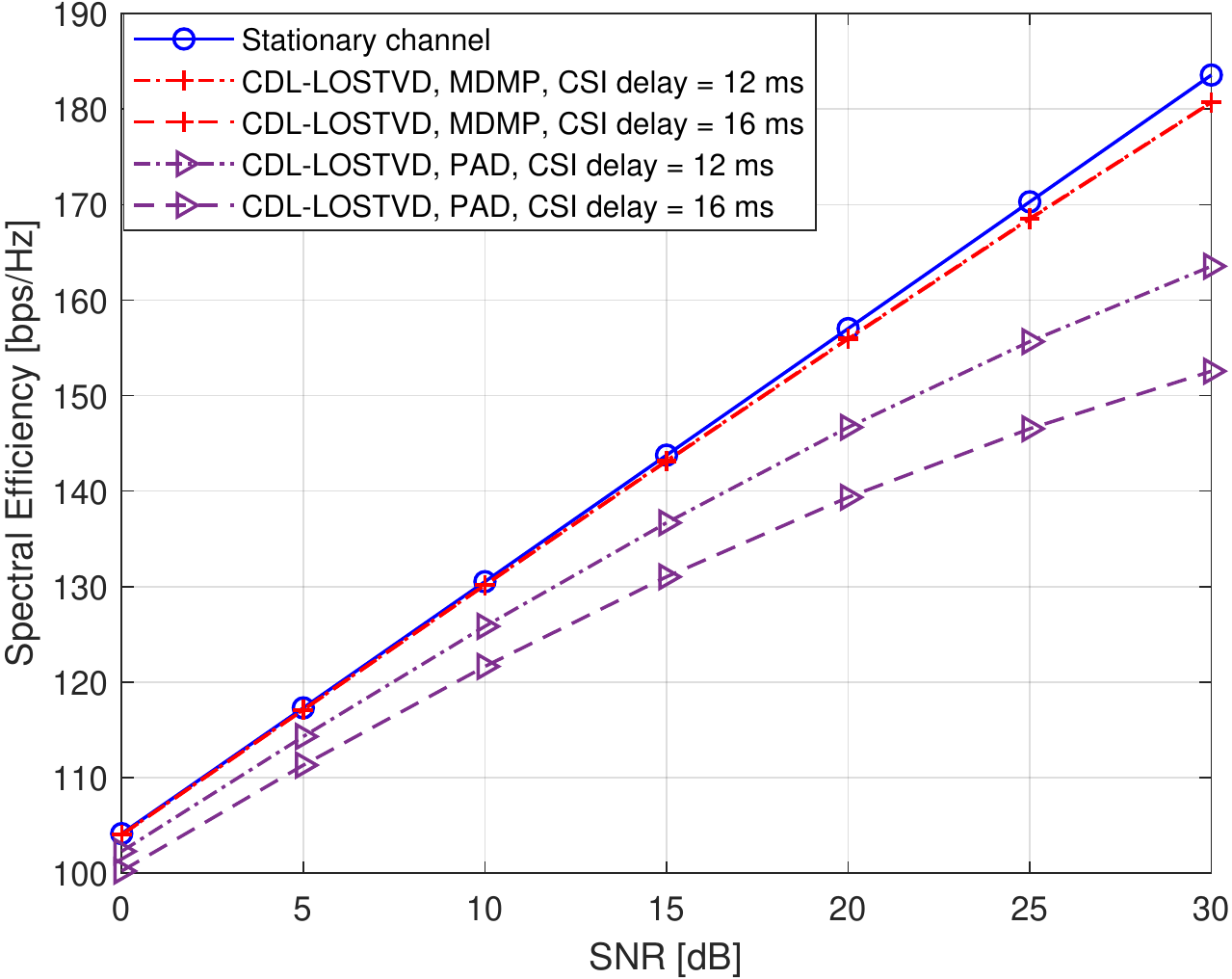}
\caption{The spectral efficiency versus SNR, UEs move at 120 km/h, the BS equipped with 64 antennas.}
\label{fig:6}
\end{figure}

\par By adding a line-of-sight (LOS) path, CDL-TVD is changed to a new model named CDL-LOSTVD, which contains 181 propagation paths. The RMS angular spreads of AOD, EOD, AOA and EOA are updated to $47.9^\circ $, $82.9^\circ $, $89.9^\circ $ and $84.6^\circ $, respectively. The spectral efficiencies of different methods in CDL-LOSTVD are shown in Fig. \ref{fig:6}. It is clear that our proposed MDMP method is capable of dealing with the effect of time-varying delay in this setting, and nearly approaches the upper-bound of the stationary setting. 
\section{Conclusion}\label{sec:conclusion}
In this paper, we addressed the challenge of mobility under the industrial channel model with time-varying path delay in massive MIMO. 
We proposed a novel multi-dimension matrix pencil channel prediction method, which estimates the angles, delay and Doppler simultaneously by exploiting the angular-time-domain and angular-frequency-domain structures of the wideband massive MIMO channel.
We also proposed a pairing method to associate the paths with the estimated angles, delay and Doppler by exploiting the super-resolution property of the angle estimation.
In the theoretical analysis, we proved that the proposed MDMP method asymptotically achieves the error-free performance with arbitrary CSI delay, providing that only two arbitrary samples are known.
We also derived a lower-bound of the number of the BS antennas for this method to work. By increasing the number of channel samples, the bound can be made lower.
Simulation results demonstrated that our proposed MDMP method in high-mobility scenarios with a large CSI delay, is very close to the performance of the stationary setting. 

\appendix
\subsection{Proof of Theorem 1}\label{appendix:arbitrary CSI delay}

For ease of exposition, we will first calculate the estimations of parameters with the observation channel samples, and then we prove that the estimations of parameters converge to the exact ones. 


Firstly, the channel parameters are estimated as follows: 
After an arbitrary CSI delay ${t_\tau }$, $t$ is updated as $t + {t_\tau }$. The channel in Eq. (\ref{Channel_UE_angular-frequency}) is rewritten as
\begin{equation}
{{\bf{\tilde H}}_u}(t + {t_\tau }) = {{\bf{A}}_u}{{\bf{C}}_u}{{\bf{\bar C}}_u}\left( {{{\bf{\bar B}}_u} \odot {{\bf{B}}_u}} \right) + {\bf{N}},
\!\!\!\!\label{Hu(t+delay)}
\end{equation}
where ${\bf{N}}$ is the sample noise matrix. The matrix ${{\bf{\bar C}}_u}$ is introduced by UE movement during ${t_\tau }$, and is expressed as
\begin{equation}
{{\bf{\bar C}}_u} = {\rm{diag}}\left\{ {\begin{array}{*{20}{c}}{e^{j2\pi {\omega _1}{t_\tau }},}&{\cdots,}&{e^{j2\pi {\omega _P}{t_\tau }}} \end{array}} \right\}.
\!\!\!\!\label{Cu_delay}
\end{equation}
The matrix ${{\bf{\bar B}}_u}$ contains multipath delay response vectors and is given by \begin{equation}
{{\bf{\bar B}}_u} = \left[ {\begin{array}{*{20}{c}}{{\bf{\bar b}}({f_1}),}&{\cdots,}&{{\bf{\bar b}}({f_{{N_f}}})} \end{array}} \right],
\!\!\!\!\label{bu_delay}
\end{equation}
where
\begin{equation}
{\bf{\bar b}}({f_{n_f}}) = {\left[ {\begin{array}{*{20}{c}}{e^{ - j2\pi {f_{n_f}} {k_{{\tau _1}}} {t_\tau }},}&{\cdots,}&{e^{ - j2\pi {f_{n_f}} {k_{{\tau _P}}} {t_\tau }}} \end{array}} \right]^T}\!\!.
\!\!\!\!\label{bf_delay}
\end{equation}
Then, the 3-D MP matrix is expressed as
\begin{equation}
\!{{\bf{\tilde G}}_u}(t + {t_\tau }) \!=\! {{\check{\bf{E}}}_{2}}(t + {t_\tau }){\bf{Y}}(t + {t_\tau }){{\check{\bf{F}}}_{2}}(t + {t_\tau }) + {\bf{N}_{{\bf{\tilde G}}_u}},
\!\!\!\!\label{Gu_delay}
\end{equation}
where ${\bf{N}_{{\bf{\tilde G}}_u}} \in {{\mathbb{C}}^{{\mu_1} \times {\mu_2}}}$ is generated by ${\bf{N}}$, and
\begin{equation}
{\bf{Y}}(t + {t_\tau }) = {\bf{Y}}{{\bf{\bar C}}_u}{\rm{diag}}\left\{ {{\bf{\bar b}}({f_1})} \right\},
\!\!\!\!\label{S_delay}
\end{equation}
\begin{equation}
\ \ \ \ \ \ \ \ \ {{\check{\bf{E}}}_{2}}(t + {t_\tau }) = {\left[\! {\begin{array}{*{20}{c}}{({{\check{\bf{E}}}_{1}})^T,}\!\!&\!{\cdots,}\!\!&\!{({{\check{\bf{E}}}_{1}}{{\bf{Z}}_{{t_\tau}}^{K - 1}})^T} \end{array}} \!\right]^T}\!,
\!\!\!\!\label{Eleft_2_delay}
\end{equation}
\begin{equation}
\!{{\check{\bf{F}}}_{2}}(t + {t_\tau }) = \left[\! {\begin{array}{*{20}{c}}{{\check{\bf{F}}}_{1}}\!&\!{\cdots,}\!&\!{{{\bf{Z}}_{{t_\tau}}^{{N_f} - K}}{{\check{\bf{F}}}_{1}}} \end{array}} \!\right],
\!\!\!\!\label{Eright_2_delay}
\end{equation}
and 
\begin{equation}
{{\bf{Z}}_{t_\tau}} = {{\bf{Z}}_\tau} {\rm{diag}}\left\{ {{\bf{\bar b}}({f_{{n_f}+1}})} \right\}{\left( {{\rm{diag}}\left\{ {{\bf{\bar b}}({f_{n_f}})} \right\}} \right)^*}.
\!\!\!\!\label{Diagonal_matrix_t_delay}
\end{equation}
As $N_h^{{\rm{BS}}} - P + 1 > L > P$, $N_v^{{\rm{BS}}} - P + 1 > R > P$ and ${N_f} - P + 1 > K > P$, $r({{\bf{\tilde G}}_u}(t + {t_\tau })) \ge P$ and ${{\bf{\tilde G}}_u}(t + {t_\tau })$ contains the angular and path delay information of the paths.
The SVD of ${{\bf{\tilde G}}_u}(t + {t_\tau })$ is ${{\bf{\tilde G}}_u}(t + {t_\tau }) = {{\bf{\tilde U}}_{u,t_\tau} }{{\bf{\tilde S}}_{u,t_\tau} }{\bf{\tilde V}}_{u,t_\tau} ^H$, where ${{\bf{\tilde U}}_{u,{t_\tau }}} = \left[ {{{\bf{\tilde u}}_{u,{t_\tau },1}}, \cdots ,{{\bf{\tilde u}}_{u,{t_\tau },{M_\tau }}}} \right]$, ${{\bf{\tilde S}}_{u,{t_\tau}} } = {\rm{diag}}\{ {{\tilde s}_{u,{t_\tau },1}}, \cdots ,{{\tilde s}_{u,{t_\tau },{M_\tau }}}\} $, and the diagonal elements of ${{\bf{\tilde S}}_{u,{t_\tau}} }$ are in non-increasing order. 
Without loss of generality, we let ${\mu_1} < {\mu_2}$, and ${M_\tau }={\mu_1}$.
We set a new threshold ${\gamma _2}$ to determine the number of the non-negligible paths as $P_1 \le P$.
Define the corresponding $P_1$ columns of ${{\bf{\tilde U}}_{u,{t_\tau}} }$ as ${{\bf{\tilde U}}_{u,s,t_\tau} }$.

We select $P_1$ columns from ${{\check{\bf{E}}}_{2}}(t + {t_\tau })$ to form ${{\check{\bf{E}}}_{s,2}}(t + {t_\tau })$. 
Since $r({{\bf{\tilde U}}_{u,s,t_\tau} })=r({{\check{\bf{E}}}_{s,2}}(t + {t_\tau }))={P_1}$, they may satisfy a mapping relationship as ${{\bf{\tilde U}}_{u,s,t_\tau} }={{\check{\bf{E}}}_{s,2}}(t + {t_\tau }){\bf{T}}_1$, where ${\bf{T}}_1 \in {{\mathbb{C}}^{{P_1} \times {P_1}}}$ is a full-rank matrix. Define the first $KRL-RL$ rows and the last $KRL-RL$ rows from ${{\bf{\tilde U}}_{u,s,t_\tau} }$ as ${{\bf{\tilde U}}_{u,s1}}$ and ${{\bf{\tilde U}}_{u,s2}}$, respectively, which satisfy
\begin{equation}
\!{{\bf{\tilde U}}_{u,s2}} \!-\! {{\bf{\tilde U}}_{u,s1}}{{\bm{\lambda }}_{t_\tau} } \!\!=\! {{\bf{J}}_1}{{\check{\bf{E}}}_{s,2}}(t \!+\! {t_\tau })({{\bf{T}}_1^{ - 1}}{{\bf{ Z}}_{t_\tau}}{\bf{T}}_1 \!-\! {{\bm{\lambda }}_{t_\tau} }) \!= \!{\bf{0}},
\!\!\!\!\label{Us1Us2_Relationship}
\end{equation}
where ${{\bm{\lambda }}_{t_\tau} } = {{\bf{\tilde U}}_{u,s1}}^\dag {{\bf{\tilde U}}_{u,s2}}={{\bf{T}}_1^{ - 1}}{{\bf{ Z}}_{t_\tau}}{\bf{T}}_1$. The matrix ${{\bm{\lambda }}_{t_\tau} }$ and the diagonal matrix ${{\bf{ Z}}_{t_\tau}}$ are similarity matrices, and share the same eigenvalues.
With the UMP method, ${{\bf{\tilde G}}_u}(t + {t_\tau })$, ${{\bf{Z}}_{t_\tau}}$, and ${{\bm{\lambda }}_{t_\tau} }$ are transformed to ${{{\bf{\tilde G}}}_{{\rm{re}}}}(t + {t_\tau })$, ${{{\bf{\hat Z}}}_{t_\tau}}$ and ${{\bf{\tilde \Psi }}_{{t_\tau }}}$, respectively, where ${{{\bf{\tilde G}}}_{{\rm{re}}}}(t + {t_\tau })$ is defined as
\begin{equation}
\begin{array}{l}
{{{\bf{\tilde G}}}_{{\rm{re}}}}(t + {t_\tau }) = {{{\bf{G}}}_{{\rm{re}}}}(t + {t_\tau }) + {{{\bf{N}}}_{{\rm{re}}}},
\end{array}
\!\!\!\!\label{Gre observation sample}
\end{equation}
and ${{{\bf{N}}}_{{\rm{re}}}}$ is generated from ${\bf{N}_{{\bf{\tilde G}}_u}}$. The diagonal matrix ${{{\bf{\hat Z}}}_{t_\tau}}$ and the matrix ${{\bf{\tilde \Psi }}_{{t_\tau }}}$ are introduced by
\begin{equation}
\begin{array}{l}
{{{\bf{\hat Z}}}_{t_\tau}} = {\rm{diag}}\{ \tan (\pi \Delta f({\tau _1}(t + {t_\tau }))), \cdots ,\\
\ \ \ \ \ \ \ \tan (\pi \Delta f({\tau _{{P_1}}}(t+{t_\tau })))\},
\end{array}
\!\!\!\!\label{Tan_t_delay}
\end{equation}
\begin{equation}
{{\bf{\tilde \Psi }}_{{t_\tau }}} = {\left[ {{\mathop{\rm Re}\nolimits} ({\bf{Q}}_{\mu_3}^H{{\bf{J}}_1}{\bf{Q}}_{\mu_1}^{}){{\bf{\tilde U}}_{s,t_\tau} }} \right]^\dag }{\mathop{\rm Im}\nolimits}({\bf{Q}}_{\mu_3}^H{{\bf{J}}_1}{\bf{Q}}_{\mu_1}^{}){{\bf{\tilde U}}_{s,t_\tau} },
\!\label{RealMatrix_CSIdelay}
\end{equation}
where ${{\bf{\tilde U}}_{s,t_\tau} }$ is the $P_1$ columns from ${{\bf{\tilde U}}_{t_\tau} }$, and ${{\bf{\tilde U}}_{t_\tau} }$ is calculated by the SVD of ${{{\bf{\tilde G}}}_{{\rm{re}}}}(t + {t_\tau })$: 
\begin{equation}
{{{\bf{\tilde G}}}_{{\rm{re}}}}(t + {t_\tau }) = {{\bf{\tilde U}}_{t_\tau} }{{\bf{\tilde S}}_{t_\tau} }{\bf{\tilde V}}_{t_\tau} ^H = {{\bf{ U}}_{t_\tau} }{{\bf{ S}}_{t_\tau} }{\bf{ V}}_{t_\tau} ^H + {{\bf{ U}}_{{\bf{N}}_{\rm{re}}} }{{\bf{ S}}_{{\bf{N}}_{\rm{re}}} }{\bf{ V}}_{{\bf{N}}_{\rm{re}}} ^H,
\!\label{Gre CSIdelay SVD}
\end{equation}
where ${{\bf{\tilde U}}_{t_\tau} } = \left[ {{{\bf{\tilde u}}_{{t_\tau },1}}, \cdots ,{{\bf{\tilde u}}_{{t_\tau },{\mu_1}}}} \right]$, ${{\bf{\tilde S}}_{{t_\tau}} } = {\rm{diag}}\{ {{\tilde s}_{{t_\tau },1}}, \cdots ,{{\tilde s}_{{t_\tau },{\mu_1}}}\} $, ${{\bf{ U}}_{t_\tau} } = \left[ {{{\bf{ u}}_{{t_\tau },1}}, \cdots ,{{\bf{ u}}_{{t_\tau },{\mu_1}}}} \right]$, and ${{\bf{ S}}_{{t_\tau}} } = {\rm{diag}}\{ {{ s}_{{t_\tau },1}}, \cdots ,{{ s}_{{t_\tau },{\mu_1}}}\} $.
Obviously, ${{\bf{\tilde \Psi }}_{{t_\tau }}}$ and ${{{\bf{\hat Z}}}_{t_\tau}}$ are similarity matrices.
With the EVD of ${{\bf{\tilde \Psi }}_{{t_\tau }}}$, we may estimate ${{{\bf{\hat Z}}}_{t_\tau}}$ by ${{{\bf{\hat Z}}}_{t_\tau}} = {\bf{W}}_{{t_\tau }} ^{ - 1}{{\bf{\tilde \Psi }}_{{t_\tau }}}{{\bf{W}}_{{t_\tau }}}$, where ${\bf{W}}_{{t_\tau }}$ contains the eigenvectors of ${{\bf{\tilde \Psi }}_{{t_\tau }}}$. After CSI delay ${t_\tau }$, the path delay ${\hat \tau _p}(t + {t_\tau })$ is estimated as
\begin{equation}
\begin{array}{l}
{\hat \tau _p}(t + {t_\tau }) = \frac{{{{\tan }^{ - 1}}\left( {{{\hat z}_{t_\tau,p}}(t + {t_\tau })} \right)}}{{\pi \Delta f}},
\end{array}
\!\!\!\!\label{Estimation_t_delay}
\end{equation}
where ${\hat z_{t_\tau,p}}(t+ {t_\tau })$ is the $p$-th estimated diagonal element of ${{\bf{\hat Z}}_{t_\tau}}$. The changing rate of path delay ${\hat k_{{\tau _p}}}$ is estimated by two different samples as
\begin{equation}
{\hat k_{{\tau _p}}} = \frac{{{{\tan }^{ - 1}}\left( {{{\hat z}_{t_\tau,p}}({t_2} + {t_\tau })} \right) - {{\tan }^{ - 1}}\left( {{{\hat z}_{t_\tau,p}}({t_1} + {t_\tau })} \right)}}{{\pi \Delta f({t_2} - {t_1})}}\!.
\!\!\!\!\label{Estimation_t_kdelay}
\end{equation}
The other parameters, i.e., Doppler, EODs and AODs can also be estimated by the similar estimation procedure of ${\hat \tau _p}(t + {t_\tau })$, and the details are omitted. 
So far, the parameters have been estimated with the observation samples.
Then, we derive the asymptotic performance of parameter estimations under Assumption 1. 

When the number of the BS antennas and the bandwidth are large, the correlation matrix of ${{{\bf{\tilde G}}}_{{\rm{re}}}}(t + {t_\tau })$ is calculated by
\begin{equation}
\begin{array}{l}
\mathop {\lim }\limits_{N_h^{{\rm{BS}}},N_v^{{\rm{BS}}},{N_f} \to \infty } {{\bm{R}}_{{{{\bf{\tilde G}}}_{{\rm{re}}}}}} = E\{ {{{\bf{\tilde G}}}_{{\rm{re}}}}(t + {t_\tau }){{{{\bf{\tilde G}}}_{{\rm{re}}}}^H}(t + {t_\tau })\}  \\
\ \ \ \ \ \ \ \ \ \ \ \ \ \ \ \ \ \ \ \ \ \ \ \ \ \ = {{\bm{R}}_{{{\bf{G}}_{{\rm{re}}}}}} + {\sigma ^2}{{\bf{I}}_{{\mu_1}}},
\end{array}
\!\!\!\!\label{Gre_correlation matrix}
\end{equation}
where ${{\bm{R}}_{{{\bf{G}}_{{\rm{re}}}}}} = E\{ {{\bf{G}}_{{\rm{re}}}}(t + {t_\tau }){{\bf{G}}_{{\rm{re}}}^H}(t + {t_\tau })\} $, ${\sigma ^2}{{\bf{I}}_{{\mu_1}}} = E\{ {{\bf{N}}_{{\rm{re}}}}{\bf{N}}_{{\rm{re}}}^H\} $, and $E\{  \cdot \} $ is the expectation over the BS antennas and bandwidth. By performing the SVD of ${{\bm{R}}_{{{{\bf{\tilde G}}}_{{\rm{re}}}}}}$, we obtain 
\begin{equation}
\mathop {\lim }\limits_{N_h^{{\rm{BS}}},N_v^{{\rm{BS}}},{N_f} \to \infty } \!\!\!\!\!\!\!{{\bf{\tilde U}}_{{t_\tau }}}{{\bf{\tilde \Sigma }}_{{t_\tau }}}{\bf{\tilde U}}_{{t_\tau }}^H \!=\! {{\bf{U}}_{{t_\tau }}}{{\bf{\Sigma }}_{{t_\tau }}}{\bf{U}}_{{t_\tau }}^H \!+\! {{\bf{U}}_{{t_\tau }}}{{\bf{\Sigma }}_{{\bf{N}}_{\rm{re}}}}{\bf{U}}_{{t_\tau }}^H,
\label{Gre_correlation SVD}
\end{equation}
where ${{\bf{\Sigma }}_{{\bf{N}}_{\rm{re}}}}={\sigma ^2}{{\bf{I}}_{{\mu_1}}}$ and
\begin{equation}
{{\bf{\tilde \Sigma }}_{{t_\tau }}}={\rm{diag}}\{ {\begin{array}{*{20}{c}}{\lvert{{\tilde s}_{{t_\tau },1}}\rvert^2,}&{\cdots,}&{\lvert{{\tilde s}_{{t_\tau },{\mu_1}}}\rvert^2} \end{array}}\},
\!\!\!\!\label{observation Sigma: delay}
\end{equation}
\begin{equation}
{{\bf{\Sigma }}_{{t_\tau }}} = {\rm{diag}}\{ {\begin{array}{*{20}{c}}{\lvert{s_{{t_\tau },1}}\rvert^2,}&{\cdots,}&{\lvert{s_{{t_\tau },{\mu_1}}}\rvert^2} \end{array}}\}.
\!\!\!\!\label{Sigma: delay}
\end{equation}
Under Assumption 1, we may obtain
\begin{equation}
\mathop {\lim }\limits_{N_h^{{\rm{BS}}},N_v^{{\rm{BS}}},{N_f} \to \infty } {{{\bf{\tilde u}}}_{{t_\tau },n}} = {{\bf{u}}_{{t_\tau },n}}{e^{j{\vartheta _n}}},
\!\!\!\!\label{Gre_SVD: U}
\end{equation}
where ${\vartheta _n}  \in [0,2\pi ]$, and
\begin{equation}
\mathop {\lim }\limits_{N_h^{{\rm{BS}}},N_v^{{\rm{BS}}},{N_f} \to \infty } \!{{\bf{\tilde \Sigma }}_{{t_\tau }}} \!=\! \mathop {\lim }\limits_{N_h^{{\rm{BS}}},N_v^{{\rm{BS}}},{N_f} \to \infty } ({{\bf{\Sigma }}_{{t_\tau }}} \!+\! {{\bf{\Sigma }}_{{\bf{N}}_{\rm{re}}}}) = {{\bf{\Sigma }}_{{t_\tau }}}.
\!\!\!\!\label{Gre_SVD: S}
\end{equation}
In other words, the $n$-th diagonal elements in ${{\bf{\tilde \Sigma }}_{{t_\tau }}}$ and ${{\bf{\Sigma }}_{{t_\tau }}}$ satisfy: $\mathop {\lim }\limits_{N_h^{{\rm{BS}}},N_v^{{\rm{BS}}},{N_f} \to \infty } {\lvert{{\tilde s}_{{t_\tau },{n}}}\rvert^2} = {\lvert{{s}_{{t_\tau },{n}}}\rvert^2 }$.
According to Eq. (\ref{Gre CSIdelay SVD}), we may obtain $\mathop {\lim }\limits_{N_h^{{\rm{BS}}},N_v^{{\rm{BS}}},{N_f} \to \infty } {{\bf{\tilde S}}_{t_\tau} } = {{\bf{S}}_{t_\tau} }$ and $\mathop {\lim }\limits_{N_h^{{\rm{BS}}},N_v^{{\rm{BS}}},{N_f} \to \infty } {P_1}=P$.
According to Eq. (\ref{Gre CSIdelay SVD}), we denote the $P$ columns of ${{{\bf{ U}}}_{{t_\tau }}}$ as ${{{\bf{ U}}}_{{s, t_\tau}}}$ and may obtain

\begin{equation}
\mathop {\lim }\limits_{N_h^{{\rm{BS}}},N_v^{{\rm{BS}}},{N_f} \to \infty } {{{\bf{\tilde U}}}_{{s, t_\tau}}} = {{{\bf{ U}}}_{{s, t_\tau}}}{\bf{\Gamma }},
\!\!\!\!\label{Gre_SVD: Us}
\end{equation}
where ${\bf{\Gamma }} = \left[ {{{\bf{\chi }}_1}, \cdots ,{{\bf{\chi }}_{{P}}}} \right] \in {{\mathbb{C}}^{P \times P}}$ and ${{\bf{\chi }}_{{n}}}$ is a ${P \times 1}$ unitary vector with the $n$-th element being ${e^{j{\vartheta _n}}}$. 
The real matrix related to path delays asymptotically converges to
\begin{equation}
\begin{array}{l}
\!\!\!\!\!\!\mathop {\lim }\limits_{N_h^{{\rm{BS}}},N_v^{{\rm{BS}}},{N_f} \to \infty }\!\!\!\!{{\bf{\tilde \Psi }}_{{t_\tau }}}\!=\!\!\!\mathop {\lim }\limits_{N_h^{{\rm{BS}}},N_v^{{\rm{BS}}},{N_f} \to \infty } \!{\left[{{\rm{Re}}({\bf{Q}}_{\mu_3}^H{{\bf{J}}_1}{\bf{Q}}_{\mu_1}^{}){{{\bf{\tilde U}}}_{s,{t_\tau }}}} \!\right]^\dag }\\
 \ \ \ \ \ \ \ \ \ \ \ \ \ \ \ \ \ \ \ \ \ \ \ \ {\rm{Im}}({\bf{Q}}_{\mu_3}^H{{\bf{J}}_1}{\bf{Q}}_{\mu_1}^{}){{{\bf{\tilde U}}}_{s,{t_\tau }}}\\
\ \ \ \ \ \ \ \ \ \ \ \ \ \ \ \ \ \ \ = \mathop {\lim }\limits_{N_h^{{\rm{BS}}},N_v^{{\rm{BS}}},{N_f} \to \infty } {\bf{\Gamma }}^{ - 1}{{\bf{\hat \Psi }}_{t_\tau}}{{\bf{\Gamma }}},
\end{array}
\!\label{RealMatrix_delay: Observation}
\end{equation}
where ${{\bf{\hat \Psi }}_{t_\tau}}$ is generated by the sample without noise:
\begin{equation}
\begin{array}{l}
{{\bf{\hat \Psi }}_{t_\tau}} = {\left[ {{\rm{Re}}({\bf{Q}}_{\mu_3}^H{{\bf{J}}_1}{\bf{Q}}_{\mu_1}^{}){{{\bf{U}}}_{s,{t_\tau }}}} \right]^\dag }{\rm{Im}}({\bf{Q}}_{\mu_3}^H{{\bf{J}}_1}{\bf{Q}}_{\mu_1}^{}){{{\bf{U}}}_{s,{t_\tau }}}
\end{array}.
\!\label{RealMatrix_delay: accurate}
\end{equation}
Obviously, ${{{\bf{\tilde \Psi }}}_{t_\tau}}$ and ${{\bf{\hat \Psi }}_{t_\tau}}$ are similar and share the same eigenvalues. In other words, despite the observation noise, the estimation of the $p$-th path delay ${\hat \tau _p}(t+{t_\tau })$ asymptotically converges to the exact value, i.e.,
\begin{equation}
\mathop {\lim }\limits_{N_h^{{\rm{BS}}},N_v^{{\rm{BS}}},{N_f} \to \infty } {\hat \tau _p}(t+{t_\tau }) = { \tau _p}(t+{t_\tau }).
\!\!\!\!\label{CSI delay: delay estimation}
\end{equation}
Likewise,
\begin{equation}
\mathop {\lim }\limits_{N_h^{{\rm{BS}}},N_v^{{\rm{BS}}},{N_f} \to \infty } {{\bf{\hat \Omega }}(t + {t_\tau })} = {{\bf{\Omega }}(t + {t_\tau })}.
\!\!\!\!\label{CSI delay: parameters estimation}
\end{equation}
With the estimated parameters in ${{\bf{\hat \Omega }}(t + {t_\tau })}$, the channel after ${t_\tau}$ is reconstructed as ${{\bf{\hat H}}_u}(t + {t_\tau })$. Then, we may obtain
\begin{equation}
\mathop {\lim }\limits_{N_h^{{\rm{BS}}},N_v^{{\rm{BS}}},{N_f} \to \infty }\!\!\!\!\!\! \frac{{\left\| {{{{{\bf{\hat h}}}_u}(t + {t_\tau })} \!-\! {{{\bf{h}}_u}(t + {t_\tau })}} \right\|_2^2}}{{\left\| {{{{\bf{h}}_u}(t + {t_\tau })}} \right\|_2^2}} = 0.
\!\!\label{CSI delay: channel error-free}
\end{equation}
\begin{prf}
Thus, Theorem 1 is proved.
\end{prf}

\subsection{Proof of Theorem 2}


We first derive the least number of samples below. 
The 2-D MP matrix ${{{\bm{\tilde \mathcal G}}_u}({n_s},{f_1})}$ generated by the observation samples, is defined as
\begin{equation}
{{{\bm{\tilde \mathcal G}}_u}({n_s},{f_1})} = {{{\bm{\mathcal G}}_u}({n_s},{f_1})} + {{\bf{N}}_{{\bm{\mathcal G}},{n_s}}},
\!\!\!\!\label{Gu 2D observation}
\end{equation}
where ${{\bf{N}}_{{\bm{\mathcal G}},{n_s}}} \in {{\mathbb{C}}^{LR \times (N_h^{{\rm{BS}}} - L + 1)(N_v^{{\rm{BS}}} - R + 1)}}$ is the noise matrix.
The 3-D MP matrix ${{{\bm{\tilde \mathcal G}}_u}({f_1})}$ is introduced by
\begin{equation}
{{{\bm{\tilde \mathcal G}}_u}({f_1})} = {{{\bm{\mathcal G}}_u}({f_1})} + {{\bf{N}}_{{\bm{\mathcal G}}}},
\!\!\!\!\label{Gu 3D observation}
\end{equation}
where ${{\bf{N}}_{{\bm{\mathcal G}}}} \in {{\mathbb{C}}^{{{\omega}_{\mu_1}} \times {{\omega}_{\mu_2}}}}$. 
Since $N_h^{{\rm{BS}}} - P + 1 > L > P$ and $N_v^{{\rm{BS}}} - P + 1 > R > P$, $r({{{\bm{\mathcal G}}_u}({n_s},{f_1})}) = P$. Therefore, ${{{\bm{\mathcal G}}_u}({n_s},{f_1})}$ and ${{{\bm{\tilde \mathcal G}}_u}({n_s},{f_1})}$ contain the angular information.
Based on the expression of ${{{\bm{\mathcal G}}_u}({n_s},{f_1})}$ in Eq. (\ref{Gu_f1_nt}), $r({{\check{\bf{E}}}_{1}}) = r({{\bf{Y}}_\omega}) = r({\bf{Z}}_{\omega_\tau}) = r({{\check{\bf{F}}}_{1}}) = P$.

Based on Eq. (\ref{Eleft2_doppler}), by applying the inequalities
\begin{equation}
r(A) \!+\! r(B) \!-\! {n_1} \!\le\! r(AB) \!\le\! \min(r(A),\!r(B)),
\!\!\!\!\label{inequality1}
\end{equation}
\begin{equation}
r(A + B) \le r(A,B) \le r(A) + r(B),
\!\!\!\!\label{inequality2}
\end{equation}
with $A \in {{\mathbb{C}}^{{m_1} \times {n_1}}}$ and $B \in {{\mathbb{C}}^{{n_1} \times {m_2}}}$, the rank of ${{\check{\bf{E}}}_{\omega,2}}$ satisfies
\begin{equation}
\begin{array}{l}
r({{\check{\bf{E}}}_{\omega,2}}) \ge r({{\check{\bf{E}}}_{1}} +  \cdots  + {{\check{\bf{E}}}_{1}}{\bf{Z}}_{\omega_\tau}^{Q - 1})\\
\ \ \ \ \ \ \ \ \ \ \ge r({{\check{\bf{E}}}_{1}}) + r((\sum\limits_{q = 1}^Q {{\bf{Z}}_{\omega_\tau}^{q - 1}} )) - P\\
\ \ \ \ \ \ \ \ \ \  = P,
\end{array}
\!\!\!\!\label{inequality1_Eleft}
\end{equation}
and $r({{\check{\bf{E}}}_{\omega,2}}) \le \min (LRQ,P) = P$, since ${{\check{\bf{E}}}_{\omega,2}} \in {{\mathbb{C}}^{LRQ \times P}}$.
Consequently, $r({{\check{\bf{E}}}_{\omega,2}}) = P$. 
Likewise, $r({{\check{\bf{F}}}_{\omega,2}}) = r({{\bf{Y}}_\omega}) = P$. 
Based on Eq. (\ref{inequality1}), $r({{\bm{\mathcal G}}_u}({f_1}))$ satisfies: 
\begin{equation}
\begin{array}{l}
r({{\bm{\mathcal G}}_u}({f_1})) \!\le\! \min(r({{\check{\bf{E}}}_{\omega,2}}),r({{\bf{Y}}_\omega}),r({{\check{\bf{F}}}_{\omega,2}})) \!= \!P,
\end{array}
\!\!\!\!\label{inequality1_3D-MP_f1_delay}
\end{equation}
\begin{equation}
\ \ \ \ \ r({{\bm{\mathcal G}}_u}({f_1})) \!\ge\! r({{\check{\bf{E}}}_{\omega,2}}) \!+\! r({{\bf{Y}}_\omega}) \!+\! r({{\check{\bf{F}}}_{\omega,2}}) \!-\! 2P \!\!=\! P.
\!\!\!\!\label{inequality2_3D-MP_f1_delay}
\end{equation}
Thus, $r({{\bm{\mathcal G}}_u}({f_1})) = P$,  which still holds if ${N_s} = Q = 2$. 

Likewise, we may derive that the smallest number of subcarriers satisfies ${N_f}=K=2$.
The estimations of parameters can be calculated by setting ${N_s}=Q=2$ and ${N_f}=K=2$. The details are omitted.
Following the similar proof procedure in Appendix \ref{appendix:arbitrary CSI delay} between Eq. (\ref{Gre_correlation matrix}) and Eq. (\ref{CSI delay: delay estimation}), we may prove 
\begin{equation}
\mathop {\lim }\limits_{N_h^{{\rm{BS}}},N_v^{{\rm{BS}}} \to \infty} {{\bf{\hat \Omega }}(t + {t_\tau })} = {{\bf{\Omega }}(t + {t_\tau })}.
\!\!\!\!\label{Two samples: parameters estimation}
\end{equation}
In other words, we may obtain
\begin{equation}
\mathop {\lim }\limits_{N_h^{{\rm{BS}}},N_v^{{\rm{BS}}} \to \infty}\!\!\!\!\!\! \frac{{\left\| {{{{{\bf{\hat h}}}_u}(t + {t_\tau })} \!-\! {{{\bf{h}}_u}(t + {t_\tau })}} \right\|_2^2}}{{\left\| {{{{\bf{h}}_u}(t + {t_\tau })}} \right\|_2^2}} = 0.
\!\!\label{Two samples: channel error-free}
\end{equation}
\begin{prf}
Thus, Theorem 2 is proved.
\end{prf}

\subsection{Proof of Proposition 1}


During the procedure of Doppler estimation, according to Eq. (\ref{NC_3D-MP_angular_frequency}), $N_h^{{\rm{BS}}}$ should satisfy $N_h^{{\rm{BS}}} \ge {F_1}(L,R)$.
Under the condition $LR = {\cal Q}$, we define a Lagrange function as
\begin{equation}
F(L,R,\lambda ) = {F_1}(L,R) + \lambda (LR - {\cal Q}).
\!\!\!\!\label{F1_lambda}
\end{equation}
Letting $\frac{{\partial F(L,R,\lambda )}}{{\partial L}} = \frac{{\partial F(L,R,\lambda )}}{{\partial R}} = \frac{{\partial F(L,R,\lambda )}}{{\partial \lambda }} = 0$, we obtain the two extreme points of $R$ and $L$: $(R_1,L_1)$ and $(R_2,L_2)$, which are shown in Eq. (\ref{R1-L1}) and Eq. (\ref{R2-L2}).
If ${\cal Q} = {\cal {N}}_s$, ${R_2}$ and ${L_2}$ are ignored. In addition, the condition of $L$ and $R$ in Eq. (\ref{Condition L R F1}) covers three extra extreme points e.g., $({R_3} = \frac{{{\cal Q}}}{2}, {L_3} = 2)$, $({R_4} = 2, {L_4} = \frac{{{\cal Q}}}{2})$, and $({R_5} = N_v^{{\rm{BS}}}, {L_5} = \frac{{{\cal Q}}}{N_v^{{\rm{BS}}}})$. 
By substituting these extreme points into ${F_1}(L,R)$, $N_t$ satisfies $N_t \ge {f_{N_v^{{\rm{BS}}},{N_s},1}}$, where ${f_{N_v^{{\rm{BS}}},{N_s},1}}$ is shown in Eq. (\ref{Proposition1_f1}).

Up to now, the sub-bound of  Doppler estimation is derived. Next, we derive the sub-bound of angle estimation.

In the procedure of angle estimation, $N_h^{{\rm{BS}}}$ and ${N_s}$ satisfy 
\begin{equation}
\begin{array}{l}
(N_h^{{\rm{BS}}} \!-\! L \!+\! 1)({N_s} \!-\! Q \!+\! 1) \!\ge\! \max \left(\! {\frac{P}{{N_v^{{\rm{BS}}} - R + 1}},1}\! \right).
\end{array}
\!\!\!\!\label{Inequality_Nh_NT}
\end{equation}
Based on the different sizes of $R$, the condition of $R$ in Eq. (\ref{Condition L R F2}) ($LQ \ge \max \left( {\frac{P}{{R - 1}},4} \right)$, and $N_v^{{\rm{BS}}} \ge R \ge 2$) is divided into four different conditions, i.e., i) $N_v^{{\rm{BS}}} - P + 1 \ge R \ge \frac{P}{4} + 1$; ii) $\min (N_v^{{\rm{BS}}} - P + 1,\frac{P}{4} + 1) \ge R \ge 2$; iii) $N_v^{{\rm{BS}}} \ge R \ge \max (N_v^{{\rm{BS}}} - P + 1,\frac{P}{4} + 1,2)$; iv) $P > R \ge \max (N_v^{{\rm{BS}}} - P + 1,2)$, which lead to  ${f_{N_v^{{\rm{BS}}},{N_s},2}}$ as in Eq. (\ref{Proposition1_f2}).
If $\frac{{\partial {F_2}(R)}}{{\partial R}}=0$, we may obtain $R_6$ and $R_7$, shown in Eq. (\ref{R6}) and Eq. (\ref{R7}).
Depending on the value of $Q$, i.e., $Q < \frac{{{N_s} + 1}}{2}$, $Q > \frac{{{N_s} + 1}}{2}$ and $Q = \frac{{{N_s} + 1}}{2}$, ${F_{2,\rm{max}}}(R)$ is obtained as in Eq. (\ref{F2max}).

Finally, the lower-bound of $N_t$ is determined by $N_t \ge \max({f_{N_v^{{\rm{BS}}},{N_s},1}}, {f_{N_v^{{\rm{BS}}},{N_s},2}})$, 
\begin{prf}
and Proposition 1 is proved.
\end{prf}

\ifCLASSOPTIONcaptionsoff
  \newpage
\fi



%

\bibliographystyle{IEEEtran}
\bibliography{references}

\begin{thebibliography}{10}
\providecommand{\url}[1]{#1}
\csname url@samestyle\endcsname
\providecommand{\newblock}{\relax}
\providecommand{\bibinfo}[2]{#2}
\providecommand{\BIBentrySTDinterwordspacing}{\spaceskip=0pt\relax}
\providecommand{\BIBentryALTinterwordstretchfactor}{4}
\providecommand{\BIBentryALTinterwordspacing}{\spaceskip=\fontdimen2\font plus
\BIBentryALTinterwordstretchfactor\fontdimen3\font minus
  \fontdimen4\font\relax}
\providecommand{\BIBforeignlanguage}[2]{{%
\expandafter\ifx\csname l@#1\endcsname\relax
\typeout{** WARNING: IEEEtran.bst: No hyphenation pattern has been}%
\typeout{** loaded for the language `#1'. Using the pattern for}%
\typeout{** the default language instead.}%
\else
\language=\csname l@#1\endcsname
\fi
#2}}
\providecommand{\BIBdecl}{\relax}
\BIBdecl

\bibitem{Li22ICC}
W.~{Li}, H.~{Yin}, and M.~{Debbah}, ``A super-resolution channel prediction
  approach based on extended matrix pencil method,'' in \emph{Proc. IEEE Int.
  Conf. Commun. (IEEE ICC)}, Seoul, South Korea, May 2022.

\bibitem{Marzetta10TCom}
T.~L. {Marzetta}, ``Noncooperative cellular wireless with unlimited numbers of
  base station antennas,'' \emph{IEEE Trans. Wireless Commun.}, vol.~9, no.~11,
  pp. 3590--3600, Nov. 2010.

\bibitem{Heath14Mag}
F.~{Boccardi}, R.~W. Heath, A.~{Lozano}, T.~L. {Marzetta}, and P.~{Popovski},
  ``Five disruptive technology directions for {5G},'' \emph{IEEE Commun. Mag.},
  vol.~52, no.~2, pp. 74--80, Feb. 2014.

\bibitem{jose2011TWC}
J.~{Jose}, A.~{Ashikhmin}, T.~L. {Marzetta}, and S.~{Vishwanath}, ``Pilot
  contamination and precoding in multi-cell {TDD} systems,'' \emph{IEEE Trans.
  Wireless Commun.}, vol.~10, no.~8, pp. 2640--2651, Aug. 2011.

\bibitem{Yin20JSAC}
H.~{Yin}, H.~{Wang}, Y.~{Liu}, and D.~{Gesbert}, ``Addressing the curse of
  mobility in massive {MIMO} with prony-based angular-delay domain channel
  predictions,'' \emph{IEEE J. Sel. Areas Commun.}, vol.~38, no.~12, pp.
  2903--2917, Dec. 2020.

\bibitem{Yin13JSAC}
H.~{Yin}, D.~{Gesbert}, M.~{Filippou}, and Y.~{Liu}, ``A coordinated approach
  to channel estimation in large-scale multiple-antenna systems,'' \emph{IEEE
  J. Sel. Areas Commun.}, vol.~31, no.~2, pp. 264--273, Feb. 2013.

\bibitem{Muller14JSTSP}
R.~R. {Müller}, L.~{Cottatellucci}, and M.~{Vehkaperä}, ``Blind pilot
  decontamination,'' \emph{IEEE J. Sel. Topics Signal Process.}, vol.~8, no.~5,
  pp. 773--786, Oct. 2014.

\bibitem{Adhikary2013}
A.~{Adhikary}, J.~{Nam}, J.~Y. {Ahn}, and G.~{Caire}, ``Joint spatial division
  and multiplexing the large-scale array regime,'' \emph{IEEE Trans. Inf.
  Theory}, vol.~59, no.~10, pp. 6441--6463, Oct. 2013.

\bibitem{Jiang15TWC}
Z.~{Jiang}, A.~F. {Molisch}, G.~{Caire}, and Z.~{Niu}, ``Achievable rates of
  {FDD} massive {MIMO} systems with spatial channel correlation,'' \emph{IEEE
  Trans. Wireless Commun.}, vol.~14, no.~5, pp. 2868--2882, May 2015.

\bibitem{Heath13JCN}
K.~T. {Truong} and R.~W. {Heath}, ``Effects of channel aging in massive {MIMO}
  systems,'' \emph{Journal of Communications and Networks}, vol.~15, no.~4, pp.
  338--351, Aug. 2013.

\bibitem{Ai21TWC}
J.~{Zheng}, J.~{Zhang}, E.~{Björnson}, and B.~{Ai}, ``Impact of channel aging
  on cell-free massive {MIMO} over spatially correlated channels,'' \emph{IEEE
  Trans. Wireless Commun.}, vol.~20, no.~10, pp. 6451--6466, Oct. 2021.

\bibitem{Larsson18TWC}
R.~{Chopra}, C.~R. {Murthy}, H.~A. {Suraweera}, and E.~G. {Larsson},
  ``Performance analysis of {FDD} massive {MIMO} systems under channel aging,''
  \emph{IEEE Trans. Wireless Commun.}, vol.~17, no.~2, pp. 1094--1108, Feb.
  2018.

\bibitem{Papazafeiropoulos15TCom}
C.~{Kong}, C.~{Zhong}, A.~K. {Papazafeiropoulos}, M.~{Matthaiou}, and
  Z.~{Zhang}, ``Sum-rate and power scaling of massive {MIMO} systems with
  channel aging,'' \emph{IEEE Trans. Commun.}, vol.~63, no.~12, pp. 4879--4893,
  Dec. 2015.

\bibitem{18GaoTWC}
J.~{Zhao}, H.~{Xie}, F.~{Gao}, W.~{Jia}, S.~{Jin}, and H.~{Lin}, ``Time varying
  channel tracking with spatial and temporal {BEM} for massive {MIMO}
  systems,'' \emph{IEEE Trans. Wireless Commun.}, vol.~17, no.~8, pp.
  5653--5666, Aug. 2018.

\bibitem{19OgawaTVT}
S.~{Uehashi}, Y.~{Ogawa}, T.~{Nishimura}, and T.~{Ohgane}, ``Prediction of
  time-varying multi-user {MIMO} channels based on {DOA} estimation using
  compressed sensing,'' \emph{IEEE Trans. Veh. Technol.}, vol.~68, no.~1, pp.
  565--577, Jan. 2019.

\bibitem{Caire20TWC}
L.~{Gaudio}, M.~{Kobayashi}, G.~{Caire}, and G.~{Colavolpe}, ``On the
  effectiveness of {OTFS} for joint radar parameter estimation and
  communication,'' \emph{IEEE Trans. Wireless Commun.}, vol.~19, no.~9, pp.
  5951--5965, Sept. 2020.

\bibitem{12TomasTVT}
T.~{Zemen}, L.~{Bernado}, N.~{Czink}, and A.~F. {Molisch}, ``Iterative
  time-variant channel estimation for 802.11p using generalized discrete
  prolate spheroidal sequences,'' \emph{IEEE Trans. Veh. Technol.}, vol.~61,
  no.~3, pp. 1222--1233, Mar. 2012.

\bibitem{19ZhuTCom}
Y.~{Zhu}, X.~{Dong}, and T.~{Lu}, ``An adaptive and parameter-free recurrent
  neural structure for wireless channel prediction,'' \emph{IEEE Trans.
  Commun.}, vol.~67, no.~11, pp. 8086--8096, Nov. 2019.

\bibitem{20YuanTWC}
J.~{Yuan}, H.~Q. {Ngo}, and M.~{Matthaiou}, ``Machine learning-based channel
  prediction in massive {MIMO} with channel aging,'' \emph{IEEE Trans. Wireless
  Commun.}, vol.~19, no.~5, pp. 2960--2973, May 2020.

\bibitem{20GesbertTWC}
M.~{Najla}, Z.~{Becvar}, P.~{Mach}, and D.~{Gesbert}, ``Predicting
  device-to-device channels from cellular channel measurements: a learning
  approach,'' \emph{IEEE Trans. Wireless Commun.}, vol.~19, no.~11, pp.
  7124--7138, Nov. 2020.

\bibitem{21GaoJSAC}
C.~{Wu}, X.~{Yi}, Y.~{Zhu}, W.~{Wang}, L.~{You}, and X.~{Gao}, ``Channel
  prediction in high-mobility massive {MIMO}: from spatio-temporal
  autoregression to deep learning,'' \emph{IEEE J. Sel. Areas Commun.},
  vol.~39, no.~7, pp. 1915--1930, July 2021.

\bibitem{92TSP}
Y.~{Hua}, ``Estimating two-dimensional frequencies by matrix enhancement and
  matrix pencil,'' \emph{IEEE Trans. Signal Process.}, vol.~40, no.~9, pp.
  2267--2280, Sept. 1992.

\bibitem{15TWC}
A.~{Gaber} and A.~{Omar}, ``A study of wireless indoor positioning based on
  joint {TDOA} and {DOA} estimation using {2-D} matrix pencil algorithms and
  {IEEE} 802.11ac,'' \emph{IEEE Trans. Wireless Commun.}, vol.~14, no.~5, pp.
  2440--2454, May 2015.

\bibitem{19DSP}
------, ``Joint estimation of time delays and directions of arrival using
  proper set of antenna elements of a high-order antenna array,'' \emph{Digit.
  Signal Process.}, vol.~94, pp. 114--124, Nov. 2019.

\bibitem{86SchmidtTAP}
R.~{Schmidt}, ``Multiple emitter location and signal parameter estimation,''
  \emph{IEEE Trans. Antennas Propag.}, vol.~34, no.~3, pp. 276--280, Mar. 1986.

\bibitem{89RoyTASSP}
R.~{Roy} and T.~{Kailath}, ``{ESPRIT}-estimation of signal parameters via
  rotational invariance techniques,'' \emph{IEEE Trans. Acoust., Speech, Signal
  Process.}, vol.~37, no.~7, pp. 984--995, July 1989.

\bibitem{3GPPR16}
3GPP, \emph{Study on channel model for frequencies from 0.5 to 100 {GHz}
  ({Release} 16)}.\hskip 1em plus 0.5em minus 0.4em\relax Technical Report TR
  38.901, available: http://www.3gpp.org.

\bibitem{91KehTSP}
K.~C. {Huarng} and C.~C. {Yeh}, ``A unitary transformation method for
  angle-of-arrival estimation,'' \emph{IEEE Trans. Signal Process.}, vol.~39,
  no.~4, pp. 975--977, Apr. 1991.

\bibitem{Yin16TSP}
H.~{Yin}, L.~{Cottatellucci}, D.~{Gesbert}, R.~R. {Müller}, and G.~{He},
  ``Robust pilot decontamination based on joint angle and power domain
  discrimination,'' \emph{IEEE Trans. Signal Process.}, vol.~64, no.~11, pp.
  2990--3003, June 2016.

\bibitem{Sun10TSP}
L.~{Sun} and M.~R. {McKay}, ``Eigen-based transceivers for the {MIMO} broadcast
  channel with semi-orthogonal user selection,'' \emph{IEEE Trans. Signal
  Process.}, vol.~58, no.~10, pp. 5246--5261, Oct. 2010.

\bibitem{Zhu19TCom}
Y.~{Zhu}, X.~{Dong}, and T.~{Lu}, ``An adaptive and parameter-free recurrent
  neural structure for wireless channel prediction,'' \emph{IEEE Trans.
  Commun.}, vol.~67, no.~11, pp. 8086--8096, Nov. 2019.

\end{thebibliography}
\end{document}